\documentclass[reqno]{amsart}
\usepackage{amsfonts}
\usepackage{amsthm}
\usepackage{amssymb}
\usepackage{dsfont}
\usepackage{amscd}
\usepackage{color}

\newtheorem{theorem}{Theorem}[section]
\newtheorem{corollary}[theorem]{Corollary}

\newtheorem{proposition}[theorem]{Proposition}
\theoremstyle{definition}

\theoremstyle{remark}
\newtheorem{remark}[theorem]{Remark}

\newtheorem*{acknow}{Acknowledgments}
\numberwithin{equation}{section}

\def\be{\begin{equation}}
\def\ee{\end{equation}}
\def\ba{\begin{eqnarray*}}
\def\ea{\end{eqnarray*}}
\def\bae{\begin{eqnarray}}
\def\eae{\end{eqnarray}}
\def\bc{\begin{center}}
\def\ec{\end{center}}

\begin{document}

\title[Raney distributions]{Raney distributions and random matrix theory}

\author{Peter J. Forrester} \address{Department of Mathematics and Statistics, The University of Melbourne, Victoria 3010, Australia}\email{p.forrester@ms.unimelb.edu.au}

\author{Dang-Zheng Liu} \address{School of Mathematical Sciences, University of Science and Technology of China, Hefei 230026, P.R. China, and
Wu Wen-Tsun Key Laboratory of Mathematics, USTC, Chinese Academy of Sciences, Hefei 230026, P.R. China}
\email{dzliu@ustc.edu.cn}

\date{\today}

%\keywords{Fuss-Catalan numbers, Moments, Free Bessel laws}

\begin{abstract}
Recent works have shown that the family of probability distributions with moments given
by the Fuss-Catalan numbers permit a simple parameterized form for their density.
 We extend this result to the Raney distribution which by definition has its
moments given by a generalization of the Fuss-Catalan numbers.
Such computations begin with an algebraic equation satisfied by the Stieltjes transform,
which we show can be derived from the linear differential equation satisfied by the characteristic
polynomial of random matrix realizations of the Raney distribution.
For the Fuss-Catalan distribution, an equilibrium problem characterizing the density is identified.
The Stieltjes transform for the limiting spectral density of the singular values squared of the matrix product formed
from $q$
inverse standard Gaussian matrices, and $s$ standard Gaussian matrices, is shown to satisfy a variant of the
algebraic equation relating to the Raney distribution. Supported on $(0,\infty)$, we show that
it too permits a simple functional form upon the introduction of an appropriate choice of parameterization.
As an application, the leading asymptotic form of the density as the endpoints of the
support are approached is computed, and is shown to have some universal features.
\end{abstract}

\maketitle
\section{Introduction}
For given $s \in \mathbb N$ the Fuss-Catalan numbers, also known as the generalized Catalan numbers, are the integer sequence
\be C_{s}(k)=\frac{1}{sk+1}\binom{sk+k}{k}, \qquad k=0,1,2,\dots \label{FCnumber}\ee
As is well known (see e.g.~\cite{LW98}) this sequence in the case $s=1$ --- traditionally referred to as the Catalan numbers --- first appeared in
the work of Euler on counting the number of  triangulations of a convex polygon consisting of $k+2$ sides. Note that each elementary triangle in
 the triangulation must contain at least one side of the $(k+2)$-gon. The other two sides of the elementary triangle therefore naturally partition the counting
 problem into two independent counting problems of the same type but involving polygons of smaller number of sides. Specifically, on one side
 the counting problem for the triangulation of an $(k_1+1)$-gon is encountered, and on the other side one has the counting problem for the
 triangulation of an $(k_2+1)$-gon, where $k_1 + k_2 = k-1$. The Catalan numbers therefore satisfy the fundamental recurrence
 \be
 C_1(k) = \sum_{k_1, k_2 \ge 0 \atop k_1 + k_2 = k -1} C_1(k_1) C_1(k_2) \label{CC}
 \ee
 valid for $k \ge 1$.

 Fuss (see e.g.~\cite{LW98}) generalized the triangulation problem of Euler to counting dissections of a convex $(sk+2)$-gon using $(s+2)$-gons.
 Here any particular $(s+2)$-gon in the dissection partitions the $(sk+2)$-gon into $(s+1)$ disjoint regions, implying the generalization of (\ref{CC})
 \be
 C_s(k) = \sum_{k_1, k_2,\dots, k_{s+1} \ge 0  \atop k_1 + k_2 + \dots +k_{s+1}  = k -1} C_s(k_1) C_s(k_2) \cdots C_s(k_{s+1}), \label{RC}
 \ee
 again valid for $k \ge 1$.

Recently, the Fuss-Catalan numbers (\ref{FCnumber})
 have appeared   in several different
contexts, for instance, product of random matrices \cite{agt,bbcc,ns}, random
quantum states \cite{cnz}, free probability and quantum groups
\cite{bbcc,mltk}. More precisely, the sequence of Fuss-Catalan numbers is the moments of some probability density $\pi_{s}$, which  is the limit spectral
distribution of  the squared singular values of random matrices with the product structures $X_1^s$ (powers of a single matrix) and
$X_{1} \cdots X_{s}$ (products of independent matrices). The $N \times N$ matrices $X_1,\dots, X_s$ are each to contain independent, identically distributed
zero mean, unit standard deviation random variables, and the squared singular values are to be divided by $N$ before the limit is taken.
 It is known that   the  explicit form of $\pi_{s}$ can be described in terms of multivariate integral representations \cite{lsw}, in terms of Meijer G-functions \cite{pz} or by using the parameterization of  the argument. This latter advance is due to Biane \cite{Ba99}, Haagerup and M\"oller \cite{HM12}, and
  Neuschel  \cite{neu}, and its further development forms one of the themes of our
 paper.

 We will consider probability densities with moments given
by a family of integer sequences generalizing the Fuss-Catalan sequence (\ref{FCnumber}).
Thus for $ p>1,  0<r\leq p$ we introduce the integer sequence
 \be  \label{raneynumber}
 R_{p,r}(k)=\frac{r}{pk+r}\binom{pk+r}{k}, \quad k=0,1,2,\dots .
 \ee
 Following \cite{pz} we refer to $R_{p,r}(k)$ as the $k$-th Raney number. To tie these numbers in with (\ref{FCnumber}) at a combinatorial
 level requires making note of a combinatorial interpretation of the latter different from that given above in terms of dissections of
 $(sk+2)$-gons. Thus suppose there are $sk$ numbers $+1$ and $k$ number $-s$ in a sequence. How many ways can the sequence members
 be arranged so that the partial sum of sequence members 1 up to $\ell$ is always non-negative for each $\ell$?
 This is a version of the so-called ballot problem (see e.g.~\cite{Re07}).
 By noting that the final member of the
 sequence must always be a $-s$, we see that the sequence with the $-s$ removed can be decomposed into $s+1$ sequences of the desired
 type, each separated by a $+1$. Hence the recurrence (\ref{RC}) holds, telling us that the answer to this counting problem is $C_s(k)$.

 As a generalization, suppose there are extra $r$ ($r > 0$) $+1$'s and it is required that the partial sums be strictly positive. With $r=1$ this is
 equivalent to the ballot problem as reviewed above, since the extra $+1$ must be appear at  the beginning of a valid ballot sequence. For
 $r \ge 2$ new ideas are needed \cite{Ra60,GKP89}. The answer is the Raney number $R_{s+1,r}(k)$. Note from the above discussion that
 $R_{s+1,1}(k)=C_{s}(k)$, which can indeed be checked from the explicit forms (\ref{FCnumber}) and (\ref{raneynumber}).

For the general case of  $ p>1,  0<r\leq p$, M{\l}otkowski \cite{mltk} has proved that $R_{p,r}(k)$ is the $k$-th moment of some probability measure $\mu_{p,r}$ (so-called Raney distributions) with compact support contained in $[0,\infty)$. In particular, explicit densities $W_{p,r}(x)$ associated with  $\mu_{p,r}$  are given in \cite{pz} for integer  $p>1$ and more generally in \cite{mpz} for rational $p>1$, both in terms of Meijer G-functions. We will show in this work
 that the parameterization  method of Biane \cite{Ba99}, Haagerup and M\"oller \cite{HM12}, and Neuschel \cite{neu} can be generalized to any real $p>1, 0<r\leq p$ and further gives explicit densities. Application of this  method to a class of probability measures supported on $(0,\infty)$ and thus not possessing finite moments, but nonetheless being intimately related to
 the sequence (\ref{raneynumber}), will be given as will independent derivations of some key polynomial equations determining the measures.

The format of the remainder of our paper is to first
review   Neuschel's derivation
  of the explicit parametrization of the Fuss-Catalan distributions, i.e.~the cases $p>1$ and $r=1$
of the Raney distributions.  In the course of this review we will see a fundamental relationship between the case $r=1$ and the general
$0 < r \le p$ case which enables us to use the results of \cite{neu} to obtain a parametrization of the Raney distributions.
As an application we give the leading term of the asymptotic expansion of the density as an endpoint of the support is approached.
We then turn our attention to realizations of the Fuss-Catalan and Raney distributions in terms of spectral densities of random matrices.
In addition to listing known examples, we add a few new cases.
We furthermore show how the polynomial equation for the resolvent can be derived from the differential equation for the characteristic equation.
For the Fuss-Catlan distribution, an equilibrium problem characterizing the density is identified involving the
logarithmic potential with image charges along rays, by making use of recent results of Claeys and Romano \cite{CR13}. 
 In the case of the spectral density for the squared singular values of a product of $q$ inverse standard Gaussian matrices, and $s$ standard
Gaussian matrices, we show how to adapt Neuschel's method to specify a parametrization of the spectral variable which allows a simple closed form
for the density. This reclaims a recent result of  Haagerup and M\"oller \cite{HM12} and furthermore allows this result to be extended. We use this to obtain the leading asymptotic form at the endpoints of the support.

\section{Parameterization of  the Raney distribution}\label{densityforraney}
The Stieltjes transform of the measure $\mu_{p,r}$, also referred to as the resolvent or Green's function,  is defined by
  \be \label{1.2a}
  G_{p,r}(z)=\int_{0}^{K_p}\frac{1}{z-x}\,d\mu_{p,r}(x) = {1 \over z} \sum_{n=0}^\infty {1 \over z^n} R_{p,r}(n), \qquad K_p=p^{p}(p-1)^{1-p},
  \ee
  where the explicit value of the upper terminal in the support $K_p$ corresponds to the radius of convergence of the series in the
  the last equality, which in turn requires that $|z| > K_p$ for its convergence.
  Crucial to our study is the fact that
$w(z):=zG_{p,r}(z)$ satisfies the algebraic equation \cite{mpz}
  \be \label{Gequation}w^{\frac{p}{r}} -zw^{\frac{1}{r}} +z=0.\ee

  Before giving its derivation, it is of interest to remark that the equation (\ref{Gequation}) is known in a different area of mathematical
  physics, namely the theory of anyons. Thus with $g$ denoting the statistical parameter, $0 \le g \le 1$ ($g=0$, $g=1$ correspond  to Bose
  and Fermi statistics respectively) it is shown in \cite{Wu94} that the statistical distribution for a single species is given by
  the average occupation number
  $$
  n_i = {1 \over W(e^{\beta (\varepsilon_i - \mu)}) + g}
  $$
  ($\beta$ is the inverse temperature and $\mu$ the chemical potential), where the function $W(x)$ satisfies the functional
  equation
  $$
  (W(x))^g (1 + W(x))^{1-g}  = x,
  $$
  with  $x := e^{\beta (\varepsilon - \mu)}$.  Introducing the transformation (see e.g.~\cite{AI99}) $w = 1 + 1/W$ the functional
  equation reads
  $$
x(w-1)=w^{1-g} %\textcolor[rgb]{1.00,0.00,0.00}{ w = 1 + x w^{1-g}},
  $$
  which is precisely (\ref{Gequation}) in the case $r=1$, with $z=x $ and $p=1-g$. Another point of interest is that
  the analytic function defined by the power series on the right hand side of (\ref{1.2a}) has been the subject of a number of
  earlier studies \cite{HN94,AI99,AI01}. In particular, with $\mathcal B_p(z) = (1/z) G_{p,1}(1/z)$, it is shown in \cite{HN94}
  that
  \be \label{2.3}
  \mathcal B_p(z)  = 1 -  \mathcal B_{1/p}\Big ( - {1 \over \sqrt[p]{-z}} \Big ).
  \ee

 To deduce (\ref{Gequation}), one first
 recalls the Lagrange inversion formula \cite{WW27}. Thus let $f(z)$ and $\phi(z)$ be analytic in a neighbourhood $\Omega$
 of $a$ and let $t$ be small enough so that $|t \phi(z)| < |z - a|$, $z \in \Omega$. The Lagrange inversion formula tells us that the equation
 \be \label{zz}
 \zeta = a + t \phi(\zeta)
 \ee
 has one solution in $\Omega$ and furthermore
  \be \label{zz1}
  f(\zeta) = f(a) + \sum_{n=1}^\infty {t^n \over n!} {d^{n-1} \over d a^{n-1}} (f'(a) (\phi(a))^n ).
  \ee
  We observe that in the case $r=1$ (\ref{Gequation}) can be rearranged to read $w = 1 + (1/z) w^p$ which is of the form
  (\ref{zz}) with $a=1$, $t=1/z$, $\phi(\zeta) = \zeta^p$. Choosing $f(\zeta) = \zeta^r$, substitution into (\ref{zz1}) shows that
  \begin{align*}
  w^r & = 1 + \sum_{n=1}^\infty {1 \over z^n n!} {d^{n-1} \over d a^{n-1}} \Big ( r a^{r-1} a^{np} \Big ) \Big |_{a=1} \\
  & = 1 + r \sum_{n=1}^\infty {1 \over z^n n!}  (np+r -1)_{n-1} \\
  & = 1 + \sum_{n=1}^\infty {1 \over z^n}  R_{p,r}(n) \: = \: zG_{p,r}(z).
  \end{align*}
  This establishes (\ref{Gequation}) in the case $r=1$, and moreover shows that
  \be \label{G1r}
  ( z G_{p,1}(z))^r = zG_{p,r}(z).
  \ee
  The latter identity together with the validity of (\ref{Gequation}) for $r=1$ establishes its validity for general $r$.

  Another viewpoint on  (\ref{Gequation}) in the case  $r=1$ is that it stems from the recurrence (\ref{RC}). Thus multiplying both sides by
  $1/z^k$ and summing over $k=1,2,\dots$ we see
  $$
  \sum_{k=0}^\infty C_s(k) z^{-k} - 1 = z^{-1} \Big ( \sum_{k=0}^\infty C_s(k) z^{-k} \Big )^s.
  $$
  Identifying $w$ with $\sum_{k=0}^\infty C_s(k) z^{-k}$, this is  (\ref{Gequation}) in the case  $r=1$.

 We know from (\ref{1.2a}) that $w(z)$ has a branch cut for $z$ on the real axis between 0 and $K_p$.
 Biane \cite{Ba99} and independently
 Neuschel  \cite{neu} sort to
 parametrize the cut by a variable $\phi$ such that (\ref{Gequation}) in the case $r=1$ permits a pair of solutions in polar form
 $w(\phi) = a(\phi) e^{i \phi}$.
 It was observed that this is possible if one uses the parametrization, a strictly decreasing function (for $0<c<1$, the function $\frac{\sin\theta}{\sin(c\theta)}$ is strictly increasing on $(0,\pi)$)
    \be \label{parameterization}x=\rho(\varphi)=\frac{(\sin p\varphi)^{p}}{\sin\varphi\,(\sin(p-1)\varphi)^{p-1}}, \qquad 0<\varphi<\frac{\pi}{p},\ee
    for then one can immediately verify that the two solutions of (\ref{Gequation}) with $r=1$ are given by
\be \label{1.5}
 \frac{\sin p\varphi}{\sin(p-1)\varphi} e^{i\varphi}\: \: \textrm{ and} \: \:  \frac{\sin p\varphi}{\sin(p-1)\varphi} e^{-i\varphi}.
 \ee

 The solutions (\ref{1.5}) have the property of both converging to the real value $p/(p-1)$ as $\varphi \to 0^+$ (i.e.~$x \to K_p^-$) and converging to the
 real value 0 as $\varphi \to \pi/p$ from below (i.e.~$x \to 0^+$). Thus they correspond to the values of $w(z)$ implied by (\ref{1.2a}) in the case
 that $z$ approaches $x$, $0 < x < K_p$ from the two sides of the cut.
 On the other hand, from the inverse formula of the Stieltjes transform, we know that the density function for the measure $\mu_{p,1}$,
 $W_{p,1}(x)$ say, is given by
 \be \label{2.9}
 W_{p,1}(x)=\lim_{\epsilon\rightarrow 0^+}\frac{1}{2i\pi}\Big(\frac{w(x-i\epsilon)}{x-i\epsilon}-\frac{w(x+i\epsilon)}{x+i\epsilon}\Big), \quad 0<x<K_p.\ee
 Consequently, upon making use of (\ref{1.5}) one has \cite{Ba99,HM12,neu}
\be  \label{1.7}
W_{p,1}(\rho(\varphi))=\frac{1}{\pi \rho(\varphi)}\frac{\sin p\varphi  \sin \varphi}{\sin(p-1)\varphi}=\frac{(\sin (p-1)\varphi)^{p-2}(\sin \varphi)^{2}}{\pi (\sin p\varphi)^{p-1}}.
\ee

Our first new result is the application of the parametrization  (\ref{parameterization}) to deduce the explicit form of the
density for the Raney distribution in the cases $p \ge r>0$. For this we observe from (\ref{G1r}) that with $x$ again parametrized by
(\ref{parameterization}), (\ref{Gequation}) permits the solutions
\be \label{1.5a}\Big ( \frac{\sin p\varphi}{\sin(p-1)\varphi} \Big )^r e^{i r \varphi}\: \:  \textrm{ and} \: \:  \Big ( \frac{\sin p\varphi}{\sin(p-1)\varphi}
\Big )^r e^{-i r \varphi}.
 \ee
These solutions have the property of each approaching 0 as $x \to 0^+$, and each approaching $p/(p-1)$ as $x \to K_p^-$. These values being real,
it follows that as for the case $r=1$ they correspond to the values of $w(z)$ on either side of the cut. Application of the inverse
Stieltjes transform formula then gives the explicit form of the density $W_{p,r}(x)$ of the measure $\mu_{p,r}(x)$ for the Raney distribution,
as we now specify.

\begin{proposition}\label{density} Let $W_{p,r}(x)$ denote the density supported
on $(0,p^{p}(p-1)^{1-p})$ with $k$th moments $R_{p,r}(k)$ of \eqref{raneynumber}. If
$$x=\rho(\varphi)=\frac{(\sin p\varphi)^{p}}{\sin\varphi\,(\sin(p-1)\varphi)^{p-1}}, \qquad 0<\varphi<\frac{\pi}{p},$$
then
\be \label{1.9}
W_{p,r}(\rho(\varphi))=\frac{(\sin (p-1)\varphi)^{p-r-1}\sin \varphi  \sin r\varphi}{\pi (\sin p\varphi)^{p-r}}, \qquad 0<\varphi<\frac{\pi}{p}.
\ee
\end{proposition}

A number of comments are in order.
\begin{remark} We observe that (\ref{1.9}) shows why it is necessary to restrict $r$ to $0 < r \le p$: only then will $W_{p,r}(\rho(\varphi))$ be
non-negative for all of the support $0 < x < K_p$ \cite{mpz}.
\end{remark}
\begin{remark} It is immediate from (\ref{raneynumber}) that $R_{p,p}(k) = R_{p,1}(k+1)$. Consequently, as observed in \cite{pz}, we must have
$W_{p,p}(x) = x W_{p,1}(x)$. Thus functional property is exhibited by Proposition \ref{density}.
\end{remark}
\begin{remark}\label{symmetricdensity} If we let  $x=y^2$, then we get a symmetric density which have $2k$-moments $R_{p,r}(k)$
\be w_{p,r}(y)=|y|W_{p,r}(y^2), \qquad y\in [-\sqrt{K_p},\sqrt{K_p}],\ee
or the standard density with variance 1
\be \widetilde{w}_{p,r}(y)=r|y|W_{p,r}(ry^2), \qquad y\in \big[ -\sqrt{K_{p}/r},\sqrt{K_{p}/r} \big].\ee
These densities, restricted to $y > 0$, are for $r=1$ the density for the singular values (rather than the singular values squared) of
the random matrices introduced in the second paragraph of the Introduction.
\end{remark}

We now turn our attention to an application of Proposition \ref{density}. We would like to use the explicit form of the density therein to
analyze its singularities near the boundary of the support, i.e.~the spectrum edges. As to be revised in subsection \ref{s3.1} below, the
case $p=2$, $r=1$ is equivalent to the Marchenko-Pastur law for the scaled density of the eigenvalues of the random matrix product
$X^* X$ ($X$ a matrix of standard Gaussians, for example). With the scaling such that the density is supported on $(0,4)$, it is immediate
that
\be \label{WX}
W_{2,1}(x) \mathop{\sim}\limits_{x \to 0^+} {1 \over \pi x^{1/2}}, \qquad
W_{2,1}(x) \mathop{\sim}\limits_{x \to 4^-} {1 \over 2 \pi} \sqrt{1 - x/4}.
\ee

The first of these behaviours distinguishes the hard edge in classical random matrix theory (see \cite[Ch.~7]{Fo10}), which in turn comes
about when the density is strictly zero for $x<0$, and the joint distribution of eigenvalues can be interpreted as the Boltzmann factor of
a one-component log-potential Coulomb gas supported on the half line $x>0$. Knowledge of this leading singular form of the density allows
the leading asymptotic $s \to \infty$ decay of the hard edge scaled (spacing between eigenvalues in the vicinity of $x=0$ of order unity) gap probabilities
for there being $k$ eigenvalues in $(0,s)$ to be computed using the Dyson log-gas heuristic \cite{FW12}. Similarly the second of the behaviours
in (\ref{WX}) distinguishes the soft edge in classical random matrix theory (see \cite[Ch.~7]{Fo10}). This edge of the support is referred to as ``soft" due to
the eigenvalue density being non-zero in the region $x>4$ before the large $N$ limit is taken. As at the hard edge, knowledge of the leading
asymptotic form of the density at the soft edge can be used in combination with the Dyson log-gas heuristic to obtain predictions for the
 leading asymptotic $s \to -\infty$ decay of the soft edge scaled (spacing between eigenvalues in the vicinity of the largest eigenvalue of order unity) gap probabilities
for there being $k$ eigenvalues in $(s,\infty)$.

Crucial to these applications of the asymptotics of the global density to the asymptotics of gap probabilities is the matching of the
former with the asymptotics of the microscopic hard and soft edge scaled densities, with the latter expanded into the bulk
\cite{GFF05}. For Gaussian Hermitian, and Wishart matrices, with real, complex and real quaternion elements this was shown
explicitly in \cite{GFF05,FFG06}. Thus, in addition to it being an intrinsic property of the densities themselves, there is much interest
in isolating the leading singular forms at the boundaries of support for the Raney distribution.
 For rational $p$ it is known from \cite{mpz,pz} that as $x\rightarrow 0^+$ the density  $W_{p,r}(x)$ is proportional to
  $x^{-(p-r)/p}$ for $r<p$, while for $r=p$  it is proportional to  $x^{1/p}$.  Proposition \ref{density} allows us to give the explicit leading
  asymptotic form of the density upon the approach of either boundary of its support.

\begin{corollary} \label{asymptoticboundaryformraney} As $x\rightarrow 0^+$, we have
\be  \label{2.16}
W_{p,r}(x) \sim \begin{cases} \frac{1}{\pi}\sin \!\frac{r\pi}{p} \, x^{-\frac{p-r}{p}},&  r<p; \\ \frac{1}{\pi}\sin \!\frac{\pi}{p} \,x^{\frac{1}{p}}, & r=p. \end{cases}
\ee
As $x\rightarrow p^{p}(p-1)^{1-p}$ from below, we have \be  W_{p,r}(x) \sim   \frac{\sqrt{2}r}{\pi} \frac{(p-1)^{p-r-3/2}}{p^{p-r+1/2}} \sqrt{1- p^{-p}(p-1)^{p-1} x}.\ee
\end{corollary}

\begin{proof}
According to (\ref{1.5a}), $x$ approaches the left boundary of support $x=0$ when $\varphi \to \pi/p$, and the precise functional form of this
approach is given by
\be \label{A2}
x\sim \Big(\frac{\sin\!p\varphi}{\sin\!\frac{\pi}{p}}\Big)^{p}.
\ee
Taking the same limit in (\ref{1.9}) shows that for $r < p$
$$
  W_{p,r}(x)\sim \frac{1}{\pi} \Big(\frac{\sin\!\frac{\pi}{p}}{\sin\!p\varphi}\Big)^{p-r} \sin\!\frac{r\pi}{p}
  $$
  while for $r=p$
$$
W_{p,r}(x)\sim \frac{1}{\pi} \sin\!p\varphi.
$$
Substituting (\ref{A2})  the first assertion follows.

The right boundary of support $x = p^p(p-1)^{1-p}$ is approached as $\phi \to 0$. Thus it follows from (\ref{1.5a}) that
\be \label{X2}
x=\frac{p^{p}}{(p-1)^{p-1}}\Big(1-\frac{p(p-1)}{2}\varphi^2+o(\varphi^2)\Big).
\ee
Taking this same limit  in (\ref{1.9}) we have
\be \label{X3}
W_{p,r}(x) \sim     \frac{r(p-1)^{p-r-1}}{\pi p^{p-r}} \varphi.
\ee
Solving (\ref{X2}) for $\phi$ and substituting in (\ref{X3}) gives the second assertion.
\end{proof}

\begin{remark}\label{R2.6}
In the case $r=1$ the asymptotic form (\ref{2.16}) was recently given in \cite[eq.~(2.16)]{Fo14}, by using the same argument as above
specialised to Biane, Haagerup and M\"oller, and Neuschel's result (\ref{1.7}).  Moreover, a matching of this asymptotic form with the asymptotic
form of the corresponding hard edge scaled microscopic density was exhibited for some values of $s$.
\end{remark}
\begin{remark}
Applying the same analysis directly to (\ref{1.5a}) shows that as $x \to 0^+$,
\be \label{B}
w(x) \sim (e^{\pm \pi i} x)^{r/p}.
\ee
In the case $r=1$, $w(x)$ is the function $\mathcal B_p(x)$ introduced above (\ref{2.3}).
Substituting (\ref{B}) shows that for $x \to \infty$, to leading order in $1/x$
$$
\mathcal B_p(x) \sim 1 + {1 \over x} .
$$
This is consistent with (\ref{1.2a}), upon recalling that $R_{p,1}(0)=1$, $R_{p,1}(1) = 1$.
 \end{remark}
\begin{remark}
A direct application of Corollary \ref{asymptoticboundaryformraney}  provides the asymptotic form of the symmetric densities given in Remark \ref{symmetricdensity}. For instance, as $y\rightarrow 0$ we have
\be
\widetilde{w}_{p,r}(y) \sim \begin{cases} \frac{1}{\pi} r^{\tfrac{r}{p}} \sin \!\frac{r\pi}{p} \, |y|^{-1+\tfrac{2r}{p}},&  r<p; \\
\frac{1}{\pi} r^{1+\tfrac{1}{p}} \sin \!\frac{\pi}{p} \, |y|^{1+\tfrac{2}{p}}, & r=p. \end{cases}
\ee

\end{remark}

\section{Some special cases}
There are a number of special cases of the Raney distribution for which the corresponding density, given in
Proposition \ref{density} in terms of our extension of Biane's and Neuschel's parametrization (\ref{1.5a}), can be written in an explicit algebraic form
using the original spectral variable. Here we list these cases.
Moreover, for a number of these a realization as the spectral density of a random matrix
ensemble is known. We give some new cases, and show in all the examples
how the corresponding algebraic
equation for the resolvent (\ref{1.2a}) can be deduced from the differential equation satisfied by the corresponding
characteristic polynomial.  For this purpose, we will see that in many cases the characteristic polynomial can be written in terms of a
generalized hypergeometric function, for which the differential equation is well known.

\subsection{$p=2$, $r=1$}\label{s3.1}
As is well known \cite{mltk,pz}, this case corresponds to both the Marchenko-Pastur law for the global density
of the squared singular values of a single random matrix
\be\label{MP1}
W_{2,1}(x)=\frac{1}{2\pi}\sqrt{\frac{4-x}{x}}, \qquad 0<x\leq 4, \ee
as well as the Wigner semi-circle law for the eigenvalues of a single real symmetric or complex Hermitian
random matrix
\be
\widetilde{w}_{2,1}(y)=\frac{1}{2\pi}\sqrt{4-y^2}, \qquad -2\leq y\leq 2.\ee

We would like to relate the algebraic equation satisfied by the resolvent
for the Wigner semi-circle law, $\tilde{G}_{2,1}(z)$ say,
to the differential equation satisfied by the characteristic polynomial. First we note that changing variables $z \mapsto z^2$ in
(\ref{1.2a}) shows that this resolvent is related to the resolvent $G_{2,1}(z)$ for the Marchenko-Pastur law by
$\tilde{G}_{2,1}(z) = z G_{2,1}(z^2)$. Recalling (\ref{Gequation}), we read off the well known fact
(see e.g.~\cite{PS11}) that $\tilde{G}_{2,1}(z)$ satisfies the
quadratic  equation
\be \label{Gz}
\tilde{G}_{2,1}^2 - z \tilde{G}_{2,1} + 1 = 0.
\ee

Consider now the averaged characteristic polynomial for a random matrix ensemble
$$
\langle \det (\lambda I - X ) \rangle_X =  \langle e^{ \log \det ( \lambda I - X )}\rangle_X.
$$
From this second expression, upon noting $ e^{ \log \det (\lambda I - X)} = \prod_{j=1}^N e^{\log ( \lambda - \lambda_j)}$, where $\{\lambda_j\}$
denote the eigenvalues of $X$, one sees that the averaged characteristic polynomial can be viewed
as the characteristic function for the linear statistic $\sum_{j=1}^N \log ( \lambda - \lambda_j) $ evaluated at $k=-i$
(recall that by definition the characteristic function for a linear statistic $\sum_{j=1}^N a(\lambda_j)$ is equal to
$ \langle \prod_{j=1}^N e^{ i k a(\lambda_j) } \rangle_X$). It is a fundamental result in classical random matrix
theory (see e.g.~\cite{Fo10,PS11}) that for large $N$, to leading order
$$
 \Big \langle \prod_{j=1}^N e^{ i k a(\lambda_j) } \Big \rangle_X \sim e^{ik N \int_J a(x) d\mu(x)},
 $$
 assuming $a(x)$ is sufficiently smooth on $J$, where $\mu(x)$ is the (normalized) spectral measure and $J$ its interval of support.
 Recently \cite{BD13} this has been extended to biorthogonal ensembles, which turns out to be the class of matrix ensembles
 giving rise to random matrix realizations of the Raney distribution to be isolated below.
 Thus it follows that for large $N$ and with $\lambda \notin J$, to leading order
 $$
 \langle \det (\lambda I - X ) \rangle_X \sim  e^{N \int_J \log (\lambda - x)  d\mu(x)}
 $$
 and in particular
 \be \label{3.4}
\lim_{N \to \infty} {1 \over N}  {d \over d \lambda} \log  \langle \det (\lambda I - X ) \rangle_X  =  \int_J {d\mu(x) \over  \lambda - x} , \qquad \lambda \notin J.
\ee
Hence we have an asymptotic relationship between the averaged characteristic polynomial and the resolvent.
This equation is stated in \cite[above (10)]{BNW13} without derivation.
As an aside, we make mention
of a recent study relating
the zeros of the averaged characteristic polynomial to the spectral density \cite{Ha13}.

Now it  is well known (see e.g.~\cite{PS11})
that the Wigner semi-circle law is the limiting spectral density for
real symmetric matrices, or  complex Hermitian matrices, with elements on and above the diagonal
 independently distributed with zero mean and variance $1/2N$.
 On the other hand, it is similarly
 well known  (see e.g.~\cite{FG06}) that the averaged characteristic polynomial for
 such matrices
  is proportional to the
Hermite polynomial $H_N(\sqrt{N/2} x)$. Furthermore, it is a classical result that  this polynomial satisfies the second order
differential equation
\be \label{Gz1}
{2 \over N} u'' - 2 x u' + 2N u = 0.
\ee
According to (\ref{3.4}),
$\tilde{G}_{2,1} =  \lim_{N \to \infty} {1 \over N} {d \over d x} \log u$. Manipulating (\ref{Gz1}) to be an equation in $u'/u$ and expanding for large
$N$, (\ref{Gz}) is reclaimed.

\subsection{$p=3$, $r=1$}
The probability density $W_{3,1}(x)$ having moments given by the Raney numbers (\ref{raneynumber}) in the case $p=3$, $r=1$, or equivalently
the Fuss-Catalan numbers (\ref{RC}) with $s=2$, first appeared in the work of Penson and Solomon \cite{ps} on quantum mechanical coherent
states. In that work it was shown
 \be
W_{3,1}(x)=\frac{1}{2^{4/3}\pi}\frac{(3\sqrt{3}+ \sqrt{27-4x})^{1/3}-(3\sqrt{3}- \sqrt{27-4x})^{1/3}}{x^{2/3}}  \label{3.5} \ee
for $0<x \leq \frac{27}{4}$, or equivalently
   \be
\widetilde{w}_{3,1}(y)=\frac{1}{2^{4/3}\pi}\frac{(3\sqrt{3}+ \sqrt{27-4y^2})^{1/3}-(3\sqrt{3}- \sqrt{27-4y^2})^{1/3}}{|y|^{1/3}}   \label{3.6} \ee
for $-\sqrt{27/4}\leq y \leq \sqrt{27/4},\  y\neq 0$.
Subsequently (\ref{3.6}) appeared as the limiting spectral density for Hermitian random matrices
\be \label{XJ}
i X^T J X, \qquad J = \mathbb I_N \otimes \begin{bmatrix} 0 & -1 \\ 1 & 0 \end{bmatrix},
\ee
where $X$ is a $2N \times 2N$ standard real Gaussian matrix. This ensemble was introduced in \cite{LSZ06} in the context of a study into a
random matrix model for disordered bosons, and later in a more mathematical context  in \cite{De10}. The eigenvalues of  (\ref{XJ}) occur in
the pairs $\pm w_j$ ($j=1,\dots,N$). It was shown in \cite{LSZ06} that the limiting density of the scaled  eigenvalues $\pm x_j := \pm w_j/\sqrt{2}N$ is
equal to (\ref{3.6}).

Here we take up of the task of computing the algebraic equation for the corresponding resolvent according to the method just given for the
resolvent of the Wigner semi-circle law. For this purpose, we first identify the characteristic polynomial for the matrices (\ref{XJ}) in terms of a
particular generalized hypergeometric function ${}_1 F_2$.

\begin{proposition}\label{P2}
We have
 \be \label{G5}
\det (\lambda \mathbb I_{2N} - i X^T J  X)   = (-2)^{-N} (2N)!\ \, {}_1 F_2\bigg ( {-N \atop 1/2, 1} \Big | {\lambda^2 \over 2} \bigg ).
\ee
\end{proposition}

\noindent
Proof. \quad We use ideas involving averages over the orthogonal group contained in \cite{FR09}, and further developed in
\cite{Fo13}. For $Y$ a square matrix, denote by $e_k(Y)$ the $k$-th elementary symmetric function (polynomial) in the eigenvalues,
$\{\lambda_j \}_{j=1,\dots,N}$ of $Y$ so that
\be \label{EY}
e_k(Y) = \sum_{1 \le j_1 < \cdots < j_k \le N} \prod_{l=1}^k \lambda_{j_l}.
\ee
We then have
\be \label{G1}
\det (\lambda \mathbb I_{2N} - i X^T J  X) =
\sum_{p=0}^{2N} \lambda^{2N-p} (-1)^p e_p(iX^T J X).
\ee
Thus our task is to compute the matrix averages $\langle  e_p(i X^T J X) \rangle_X$, where $X$ is drawn from the
set of $2N \times 2N$ real standard  Gaussian random matrices.

Using the formula \cite[(3.2)]{FR09} we have the simplification
\be \label{G2}
\langle  e_p(i X^T J X) \rangle_X = {e_p(iJ) \over e_p(\mathbb I_{2N})} \langle  e_p( X^T  X) \rangle_X .
\ee
As a consequence of the eigenvalues of $iJ$ being equal to $\pm 1$, each with multiplicity $N$,
we have $\sum_{p=0}^{2N} z^{p}  e_p(iJ)  =  (1 -  z^2)^N$, and so
\be \label{G3}
 e_p(iJ) = \left \{  \begin{array}{ll} 0, & p \: \: {\rm odd} \\[.2cm]
 \displaystyle (-1)^{p/2}  \binom{N}{p/2}, & p \: \: {\rm even}. \end{array} \right.
 \ee
 Furthermore, we read off from \cite[(3.9)]{FR09} that
 \be \label{G4}
 {1 \over e_p(\mathbb I_{2N})} \langle  e_p(X^T  X) \rangle_X = 2^{p/2} \prod_{j=1}^p (N - (j-1)/2).
 \ee
 Substituting (\ref{G3}) and (\ref{G4}) into (\ref{G2}), then substituting the result in (\ref{G1}) with $p$ replaced by $2N - p$, we see after
 minor manipulation that
 $$
 \det (\lambda \mathbb I_{2N} - i X^T J  X) = (-2)^{-N} (2N)! \sum_{p=0}^N {(-N)_p \over p! (1/2)_p (1)_p} (\lambda^2/2)^p.
 $$
 This is precisely (\ref{G5}). \hfill $\square$

 \medskip
 To make use of Proposition \ref{P2}, we require the standard fact that the generalized hypergeometric function
 ${}_p F_q({a_1,\dots,a_p \atop b_1,\dots, b_q} | x)$ satisfies the differential equation
 \be \label{DpFq}
 x \prod_{n=1}^p \Big ( x {d \over dx} + a_n \Big ) f =  x {d \over dx} \prod_{n=1}^q \Big ( x {d \over dx} + b_n - 1 \Big ) f .
 \ee

 The Green's function corresponding to  (\ref{3.6}), $\tilde{G}_{3,1}(z)$ say, is related to the Green's function for  (\ref{3.5}) by
 $\tilde{G}_{3,1}(z) = z G_{3,1}(z^2)$, and thus recalling (\ref{Gequation}) must satisfy
 \be \label{TG}
 z\tilde{G}_{3,1}^3 - z \tilde{G}_{3,1} + 1 = 0  .
 \ee
 We know the Green's function is related to the characteristic polynomial by (\ref{3.4}). Denoting the characteristic polynomial by
 $p(\lambda)$, it follows from Proposition \ref{P2} and (\ref{DpFq}) that
  \be  \label{BV}
 {\lambda^2 \over 2} \Big ( {\lambda  \over 2} {d \over d \lambda} -  N \Big )  p =  {\lambda \over 2}
  {d \over d\lambda}  \Big (  {\lambda \over 2}  {d \over d\lambda} \Big )
   \Big (  {\lambda \over 2}  {d \over d\lambda}  - {1 \over 2} \Big )  p .
\ee
By expressing higher derivatives of $p$ in terms of the logarithmic derivative, and recalling from (\ref{3.4}) that the latter is proportional
to $N$ for large $N$, we see that to leading order in $N$
\be \label{HN}
{p^{(k)} \over p} \sim \Big ( {p' \over p} \Big )^k
\ee
(in the case $k=2$ this equation has already been used in going from (\ref{Gz1}) to (\ref{Gz})). Using this fact, now with
$k=3$, and furthermore replacing $\lambda$ by $\sqrt{2}N z$\, and recalling (\ref{3.4}) with $N$ by $2N$ we see that (\ref{BV}) reduces to (\ref{TG}) in the
$N \to \infty$ limit.

If we consider (\ref{3.5})  rather than (\ref{3.6}) there is another random matrix interpretation to the Raney distribution with parameters
$p=3$, $r=1$. This has already been mentioned in the third paragraph of the Introduction: it gives the limit spectral distribution of
the squared singular values of the random matrix power $X^2$, or random matrix product $X_1 X_2$, with $X, X_1, X_2$
standard real Gaussian random matrices (for example) \cite{pz,Zh13}. This in turn is a special case of the result \cite{pz} that the Raney distribution
with parameters $p=s+1$, $r=1$ gives the limit spectral distribution of
the squared singular values of the random matrix power $X^s$, or random matrix product $X_1 \cdots X_s$. For this latter problem, we can also
deduce the polynomial equation (\ref{1.2a}) satisfied by the resolvent from knowledge of the differential equation for the corresponding averaged
characteristic polynomial, as done in the above calculations.

Let $X_i$, $i=1,\dots,s$ be independent standard complex Gaussian $N \times N$ matrices. The averaged characteristic polynomial has
been shown in \cite{AIK13,AKW13} to be equal to
\be
(-1)^N (N!)^s \, {}_1 F_s \Big ( {-N \atop 1,\dots,1} \Big | \lambda \Big ).
\ee
We thus read off from (\ref{DpFq}) that the characteristic polynomial satisfies the differential equation
\be
\lambda \Big ( \lambda {d \over d\lambda} - N  \Big ) p =  \lambda  {d \over d\lambda} \prod_{n=1}^s  \lambda {d \over d\lambda}  p.
\ee
Changing variables $\lambda = N^s z$, replacing $p'(z)/p(z)$ by $N G(z)$ as is consistent with  (\ref{3.4}), and expanding for
large $N$ using  (\ref{HN}) we deduce that
\be
z( z G(z) - 1) =  (zG(z))^{s+1}.
\ee
With $zG(z) = w(z)$ this is precisely  (\ref{1.2a}) in the case $p = s+1, r=1$.

\subsection{$p=3$, $r=2$} \label{S3.3}Penson and \.{Z}yczkowski \cite{pz} have shown that with these parameters, the density for the
Raney density takes on the explicit form
\be \label{ND0}
W_{3,2}(x)=\frac{1}{2^{5/3}3^{1/2}\pi}\frac{(3\sqrt{3}+ \sqrt{27-4x})^{2/3}-(3\sqrt{3}- \sqrt{27-4x})^{2/3}}{x^{1/3}}  \ee
for $0<x \leq \frac{27}{4}$, or equivalently
 \be \label{ND}
\widetilde{w}_{3,2}(y)=\frac{ |y|^{1/3}}{2\sqrt{3}\pi}\Big((3\sqrt{3}+ \sqrt{27-8y^2})^{2/3}-(3\sqrt{3}- \sqrt{27-8y^2})^{2/3}\Big)    \ee
for $-\sqrt{27/8}\leq y \leq \sqrt{27/8}$.
They do not give any random matrix realization. However, it is possible to find (\ref{ND}) in the thesis of Nadal \cite[eq.~(6.118)]{Na11}, where it
appears as the global density for the Pearcey process. The Pearcey process corresponds to the eigenvalue density for the critical case of the
Gaussian unitary ensemble with two sources, symmetrically place about the origin  \cite{BH98}. Explicitly, the random matrices for this ensemble are
of the form  $X + X_0$,
where $X$ is a member of the Gaussian unitary ensemble (matrices $(Y+Y^\dagger)/2$ where $Y$ is a standard complex Gaussian matrix),
and $X_0$ is a diagonal matrix with half the diagonal entries equal to $a$ and the other half equal to $-a$. The parameter $a$ is be related to $N$ in
such a way that for large $N$ the spectrum vanishes as a power law as the origin is approached --- this is the critical case described by the
Pearcey process \cite{BK07}.

In this setting one encounters in \cite{BK07} the cubic equation
\be \label{bp}
\xi^3 - x \xi^2 + x = 0.
\ee
Introducing the variable $h := 1 - \xi/x$ this reads
\be \label{bp1}
h^3 - 2  h^2 +   h - 1/x^2 = 0.
\ee
We would like to relate (\ref{bp}) to an equation satisfied by $G_{3,2}(z)$. First, we note from (\ref{Gequation}) that $w(z) := z G_{3,1}(z)$ satisfies the cubic
equation $w^3 - z w + z = 0$. Taking $z$ to the RHS and squaring gives $w^6 - 2 z w^4 + z^2 w^2 = z^2$. On the other hand, (\ref{G1r}) gives that
$w^2 = z G_{3,2}(z)$ and so
\be \label{bp2}
(z G_{3,2})^3 - 2 z (z G_{3,2})^2 + z^2 ( z G_{3,2}) - z^2 = 0.
\ee
We recognize this equation as identical to (\ref{bp1}) with $G_{3,2} = h$ and $z =x^2$.
Introducing $\tilde{G}_{3,2}(z) = z G_{3,2}(z^2)$, in keeping in going from (\ref{ND0}) to (\ref{ND}), (\ref{bp2}) reads
\be \label{bp3}
(\tilde{G}_{3,2})^3 - 2 z (\tilde{G}_{3,2})^2 + z^2  \tilde{G}_{3,2}  - z = 0.
\ee
In the remainder of this section, we will show how the
cubic equation (\ref{bp3}) can be deduced from the differential equation satisfied by the characteristic polynomial for the Gaussian unitary ensemble
with two sources.

The eigenvalue distribution for the Gaussian unitary ensemble with a source is an example of a biorthogonal ensemble, with the corresponding
biorthogonal system being that which specifies the so-called type II multiple Hermite polynomials.
General theory \cite{BK04,DF08} tells us that the characteristic polynomial must then be given in terms of the latter. In the case of interest, where the
source is specifed by a diagonal matrix $X_0$ with diagonal entries taking on only two possible values, the multiple Hermite polynomials
are indexed by two non-negative integers $m$ and $n$ say.  It is known that these polynomials satisfy a third order linear differential equation
\cite{Is05,CV06,FVZ12}, which with source parameters denoted by $c_1$ and $c_2$ we read off from \cite{FVZ12} to be given by
\begin{align}\label{pD}
p'''(x) + (c_1 + c_2 - 4x) p''(x) + &(c_1(c_2 - 2x) + 2(m+n-1-c_2 x + 2 x^2)) p'(x)  \nonumber \\
& + 2 (c_1 m + c_2 n - 2(m+n)x) p(x) = 0.
\end{align}

More specifically, we want $X_0$ to consist of $N$ values $a$ and $N$ values $-a$.
The characteristic polynomial will then be proportional to the type II multiple Hermite polynomial with indices $m=n=N$.
Specifying (\ref{pD}) as such, choosing $c_1 = - c_2 = 2 \sqrt{N}a$ and scaling $x \mapsto \sqrt{N} x$, we obtain that the
characteristic polynomial satisfies the differential equation
\be
{1 \over N^2} p'''(x) - {4x \over N} p''(x) + 4 \Big ( - a (a+x) + ( 1 - {1 \over 2N} + x^2) \Big ) p'(x) - 8 N x p(x) = 0.
\ee
To be consistent with (\ref{3.4}) we replace $p'(x)/p(x)$ with $2N g(x)$ (the factor of 2 comes about because there are a total
of $2N$ eigenvalues), then expand for large $N$ using  (\ref{HN}). This gives that $g(x)$ satisfies the cubic equation
\be \label{bp4}
g^3 - 2x g^2 + (1 - a^2 + x^2)g - x = 0.
\ee
Setting $a=1$, which corresponds to the critical case, we see that this is precisely the equation (\ref{bp3}).

\begin{remark}
Replacing $g$ in (\ref{bp4}) by $\xi$ according to $\xi = x - g$ we see that $\xi$ satisfies
\be \label{bp5}
\xi^3 - 2x \xi^2 - (a^2  - 1) \xi + x a^2 = 0
\ee
This equation appears in \cite{BK07} as the generalization of (\ref{bp}) for general (non-critical) value of the rescaled source parameter $a$.
\end{remark}

\begin{remark}
Complex chiral matrices have the structure
\be \label{XX}
\begin{bmatrix} 0_N & X \\
X^* & 0_N \end{bmatrix},
\ee
where $X$ is a complex standard Gaussian matrix. The eigenvalues of such matrices occur in pairs $\pm \lambda_j$ $(\lambda_j > 0,
j=1,\dots,N)$. A source is obtained by adding to (\ref{XX}) a matrix of the same structure as (\ref{XX}) but with $X$ replaced by $X_0=c \mathbb I_N$,
$c > 0$. Furthermore the matrix $X$ is scaled by a time variable $t$, and so in particular
it vanishes for $t=0$. The distribution of the squared eigenvalues for this model, or equivalently the distribution of the eigenvalues of the shifted mean
Wishart matrices
\be \label{XX1}
(X+X_0)^*(X + X_0),
\ee
 is identical to that for $N$ non-intersecting squared Bessel processes all starting from
$x=c$ at time $t=0$ \cite{KO01}. By scaling $c$ with $N$, there is a well defined global density, and there is a critical value of this parameter
for which the density touches the origin for the first time.

Moreover, the (squared) eigenvalue distribution is an example of a biorthogonal ensemble, and the corresponding characteristic polynomial is an example
of a type II multiple Laguerre polynomial. The latter satisfies a third order differential equation \cite{CV03}. In fact the equation satisfied by the
large $N$ form of the logarithmic derivative --- which we know from (\ref{3.4}) gives the resolvent --- has been deduced from this
\cite{KMW09} and is given by
$$
z = {1 \over w (1 - w)^2}.
$$
This is just a rewrite of  (\ref{bp2}) with $w = G_{3,2}$.
\end{remark}

\begin{remark}\label{remarkcritical}
Let $X_0$ be defined as in (\ref{XX1}), let $X$ be a standard complex Gaussian matrix, and similarly let $Y_1,\dots,Y_s$ be standard
complex Gaussian matrices, all of size $N \times N$. What can be said about the scaled global density of the squared singular values
of the product $(X + X_0) Y_1 \cdots  Y_s$ in the case that the
parameter $c$ in $X_0$ is scaled to correspond to the critical case? As just discussed, the scaled density for the squared singular values
of $X + X_0$ in the critical case is the Raney distribution $W_{3,2}$, while as rederived in the above subsection, this scaled density for
the product $Y_1 \cdots  Y_s$ is the Raney distribution $W_{s+1,1}$.  In the language of free probability (see e.g.~\cite{bbcc}) we are
seeking the value of the free multiplication convolution $W_{3,2}\boxtimes W_{s+1,1}$. It is known \cite[eq.~(4.14]{mltk}) that in general
\be \label{WM}
 W_{p,r}\boxtimes W_{s+1,1}=W_{p+rs,r},
 \ee
 and so $W_{3,2}\boxtimes W_{s+1,1} = W_{3+2s,2}$. It is worth stressing that the leading form for $W_{3+2s,2}$ near zero, being given by
 $$\frac{1}{\pi}\sin \!\frac{2\pi}{3+2s} \, x^{ -1+\frac{1}{1+s+1/2}
 },$$
 is different from the Fuss-Catalan distribution, which may indicate some new universality phenomenon at the hard edge.
 \end{remark}

\subsection{$p=3/2$, $r=1/2$}
We see from (\ref{Gequation}) that this choice of parameters, as with the previous two considered above, gives a cubic equation for the
resolvent. In keeping with this is the explicit form of the density
\cite{mpz}  \be \label{bu1}
W_{3/2,1/2}(x)=\frac{1}{2^{5/3}3^{1/2}\pi}\frac{(3\sqrt{3}+ \sqrt{27-4x^2})^{2/3}-(3\sqrt{3}- \sqrt{27-4x^2})^{2/3}}{x^{2/3}}  \ee
for $0<x \leq \sqrt{27/4}$ or equivalently
\be
\widetilde{w}_{3/2,1/2}(y)=\frac{1}{4\sqrt{3}\pi}\frac{(3\sqrt{3}+ \sqrt{27-2y^4})^{2/3}-(3\sqrt{3}- \sqrt{27-2y^4})^{2/3}}{|y|^{1/3}}    \ee
for $-\sqrt[4]{27/2} \leq y \leq \sqrt[4]{27/2}, y\neq 0$.
Furthermore, it was observed in \cite{mpz} that  (\ref{bu1}) occurs in random matrix theory as the so called Bures distribution \cite{sz}.
This is realized as the global eigenvalue density of the random matrices $(1+U)XX^{*}(1+U^*)$, where $X$ is a complex standard Gaussian
matrix, and $U$ is a random unitary matrix with Haar measure \cite{osz}. We would like to use this realization to deduce first the explicit form
of the averaged characteristic polynomial, and then from that the cubic equation for the resolvent.

\begin{proposition}\label{P3}
For $X$ and $U$ members of the ensembles as specified,
we have
 \be \label{G5b}
\langle \det (\lambda \mathbb I_{N} - (1+U)XX^{*}(1+U^*)) \rangle_{X,U}  =   (-1)^N (N+1)!  \, {}_2 F_2\bigg ( {-N , N+2 \atop 1,3/2} \Big |  4 \lambda \bigg ).
\ee.
\end{proposition}

\noindent
Proof. \quad We proceed as in the proof of Proposition \ref{P2}. Thus we write
\be \label{G1a}
\det (\lambda \mathbb I_{N} - (1+U)XX^{*}(1+U^*) ) =
\sum_{p=0}^{N} \lambda^{N-p} (-1)^p e_p((1+U)XX^{*}(1+U^*)),
\ee
showing that our task is to compute the matrix averages $\langle  e_p((1+U)XX^{*}(1+U^*)) \rangle_{X,U}$, where $X$ is drawn from the
set of $N \times N$ complex standard  Gaussian random matrices, and $U$ from the set of unitary matrices with Haar measure.

Analogous to (\ref{G2}) we have
\be \label{G2a}
\langle  e_p((1+U)XX^{*}(1+U^*)) \rangle_{X,U} = { \langle e_p((1+U)(1+U^*) \rangle_U  \over e_p(\mathbb I_{N})} \langle  e_p( X^T  X) \rangle_X .
\ee
Furthermore, we read off from \cite[(3.13)]{FR09} that
 \be \label{G4a}
 {1 \over e_p(\mathbb I_{N})} \langle  e_p(X^T  X) \rangle_X =   \prod_{j=1}^p (N - (j-1)).
 \ee
To compute the average over $U$, we first note that
$$
\sum_{p=0}^N \zeta^p \langle e_p((1+U)(1+U^*) \rangle_U =  \Big  \langle  \det \Big ( \mathbb I_{N} + \zeta (1+U)(1+U^*) \Big )  \Big \rangle_U
$$
In terms of the eigenvalues of $U$, $\{z_j\}_{j=1,\dots,N}$ say with $|z_j|=1$, the determinant factorizes to read
$$
  \det \Big ( \mathbb I_{N} + \zeta (1+U)(1+U^*) \Big ) = \prod_{j=1}^N \Big ( 1 + 2 \zeta + \zeta (z_j + {1 \over z_j}) \Big ).
$$
Recalling now that any average over the unitary group involving a product of the individual eigenvalues can
be written as a Toeplitz determinant (see e.g.~\cite[eq.~(5.76)]{Fo10}), and noting that in the present case the
elements in position $(jk)$ of the latter are given by
$$
{\rm CT} \,  \Big ( 1 + 2 \zeta + \zeta (z + {1 \over z } ) \Big ) z^{j-k} =
\left \{ \begin{array}{ll} \zeta, & |j-k|=1 \\
1 + 2 \zeta, & j=k \\
0, & {\rm otherwise} \end{array} \right.,
$$
where CT denotes the constant term in the Laurent expansion with respect to $z$,
we see that our task is reduced to one of evaluating a symmetric tridiagonal Toeplitz determinant  which is constant down both the diagonal
and its two neighbours.

We can verify that with $B$ the symmetric $N \times N$ matrix with entries $x$ in positions with $|j-k|=1$ and entries 0 elsewhere
$$
\det ( \mathbb I_N + B) = \sum_{p=0}^{[N/2]} \binom{N-p}{p} (-1)^p x^{2p}.
$$
This allows us to deduce that
$$
  \Big  \langle \det \Big ( \mathbb I_{N} + \zeta (1+U)(1+U^*) \Big )\Big  \rangle_{U} =  \sum_{p=0}^{[N/2]} \binom{N-p}{p}
  (-1)^p \zeta^{2p} (1 + 2 \zeta)^{N-2p}.
$$
Extracting the coefficient of $\zeta^p$ from this gives a sum over binomial coefficients which can be
evaluated, and we conclude
\be   \label{G4b}
\langle  e_p((1+U)XX^{*}(1+U^*)) \rangle_{X,U}   = { (2N - p + 1)! \over
(2N-2p + 1)! p!}.
\ee

Substituting (\ref{G4a}) and (\ref{G4b}) in (\ref{G1a}) we obtain
\be \label{G7}
 \Big  \langle \det (\lambda \mathbb I_{N} - (1+U)XX^{*}(1+U^*)) \Big  \rangle_{X,U} = (-1)^N N! \sum_{q=0}^N {(N+q+1)! \over q! (2q+1)! (N-q)!}  (-\lambda)^q.
\ee
Noting that
$$
{N! \over (N-q)!} = (-1)^q (-N)_q, \quad
(2q+1)! = 2^{2q}(1)_q (3/2)_q, \quad
(N+q+1)! = (N+1)! (N+2)_q
$$
we see that (\ref{G7}) is equivalent to (\ref{G5b}).
\hfill $\square$

\medskip
Replacing $\lambda$ by $Nx/4$, it follows from (\ref{G5b}) and (\ref{DpFq}) that to leading order in $N$
$$
x \Big ( \Big (x {d \over dx} \Big )^2 - N^2 \Big ) p = \Big (  x {d \over dx} \Big )^3 p.
$$
Now making use of (\ref{HN}) in the cases $k=2,3$ and recalling (\ref{3.4}) it follows that the limiting resolvent,
$g$ say, satisfies the cubic equation
$$
x ( (xg)^2 - 1) = (x g)^3,
$$
which with $xg = w$ is precisely (\ref{Gequation}) in the case $p=3/2$, $r=1/2$.

\begin{remark}\label{remarkbures}
The eigenvalues of the random matrix $(1+U) X X^* (1 + U^*)$ are equal to the squared singular values of $(1+U)X$. What can we
say about the squared singular values of $(1+U) X X_1 \cdots X_s$, where $X$ and each $X_i$ is a standard complex Gaussian matrix, and
$U$ is a random unitary matrix chosen with Haar measure? We are asking for $W_{3/2,1/2} \boxtimes W_{s+1,1}$. Applying (\ref{WM}) tells
us that this is equal to $ W_{(3+s)/2,1/2}$. In particular, we have \cite{mpz}
\be \label{W212}
W_{2,1/2}(x)=\frac{\sqrt{2-\sqrt{x}}}{2\pi x^{3/4}}, \qquad 0<x<4.
\ee
\end{remark}

\section{Equilibrium problem}\label{S3}
In Section \ref{S3.3} a realization of the Raney distribution with parameters $p=3$, $r=2$ as the eigenvalue density for the
Pearcey process was given. It was commented that the corresponding symmetric density $\tilde{w}_{3,2}(y)$ as specified
by (\ref{ND}) can be found in the thesis of Nadal \cite{Na11}. In the latter it is furthermore shown that $\tilde{w}_{3,2}(y)$ is the density
$\rho_{(1)}(x)$ which minimizes the energy functional
\begin{align}\label{21.1}
\tilde{E}_{3,2}[\rho_{(1)}(x)] =&
\int_0^L (x + 1/\sqrt{2})^2 \rho_{(1)}(x) \, dx  \nonumber \\
& -{1 \over 2}
\int_0^L dx \int_0^L dx' \,\rho_{(1)}(x) \rho_{(1)}(x') \log \Big ( | x - x'| |x^2 - (x')^2| \Big ),
\end{align}
subject to the normalization $\int_0^L \rho_{(1)}(x) \, dx = 1/2$.
The functional (\ref{21.1}) is very revealing when compared against the energy functional known for the
Raney distribution with $p=3$, $r=1$ using squared variables, which we know has density (\ref{3.6}). It is shown
in \cite{LSZ06} that
 $\tilde{w}_{3,1}(y)$ is the density
$\rho_{(1)}(x)$ which minimizes the energy functional
\begin{align}\label{21.2}
\tilde{E}_{3,1}[\rho_{(1)}(x)] = &
\int_0^L x \rho_{(1)}(x) \, dx   \nonumber \\
& -
{1 \over 2}
\int_0^L dx \int_0^L dx' \,\rho_{(1)}(x) \rho_{(1)}(x') \log \Big ( | x - x'| |x^2 - (x')^2| \Big ),
\end{align}
subject to the normalization $\int_0^L \rho_{(1)}(x) \, dx = 1/2$.

We thus see that both (\ref{21.1}) and (\ref{21.2}) are special cases of the family of energy functionals
\begin{align}\label{21.3}
\tilde{E}(V;\theta)[\rho_{(1)}(x)] = &
\int_0^L V(x) \rho_{(1)}(x) \, dx   \nonumber \\
&-{1 \over 2}
\int_0^L dx \int_0^L dx' \,\rho_{(1)}(x) \rho_{(1)}(x') \log \Big ( | x - x'| |x^\theta  - (x')^\theta | \Big ),
\end{align}
for $\theta = 2$ and particular $V$. In the case $\theta \in \mathbb Z^+$ by writing
$ |x^\theta  - (x')^\theta | = \prod_{k=0}^{\theta-1}  | x - \omega^{k} x'|$ where $\omega = e^{2 \pi i / \theta}$ we can
interpret this as the Boltzmann factor for a log-potential Coulomb gas on the half line at inverse temperature
$\beta = 2$, with image particles of charge $+1/2$ along rays in the direction of $\omega^k$ for $k=1,\dots,\theta - 1$
(see \cite[\S 3.1.4]{Fo10} for a related interpretation of an energy functional).
The significance of (\ref{21.3}) is that in the case that $V(x)$ is such
that $\rho_{(1)}(x)$ is supported on $(0,L)$ (one-cut assumption with hard edge) or $(L_0,L)$ with $0<L_0<L$ the
equilibrium problem for (\ref{21.3}) has been recently solved by Claeys and Romano \cite{CR13} for all $\theta \ge 1$.
Examination of their results reveals an intimate relationship with the Raney distribution, and in fact allows us to
specify the equilibrium problem for all Raney distributions with $r=1$ (Fuss-Catalan case) and also to provide a
qualitative understanding of the cases $r \in \mathbb Z^+$ for $r < p$.

First, some notation is needed. With $\theta \ge 1$ the same parameter as in (\ref{21.2}), define
\begin{equation}\label{JT}
J(s) = (s + 1) \Big ( {s + 1 \over s} \Big )^{1/\theta} \theta.
\end{equation}
As done in \cite{CR13}, it is easy to verify that there are two complex conjugate curves $\gamma_\pm$, in the upper
and lower half plane respectively, which join $s=-1$ and $s=1/\theta$ and are such that $J(s)$ is real for $s \in \gamma_\pm$.
Define $I_\pm(x) \in \gamma_\pm$ by the requirement that
\begin{equation}\label{JI}
J(I_\pm(x)) = x.
\end{equation}
And with $\rho_{(1)}(x)$ the density minimizing (\ref{21.3}), and $H_\theta = \{ z \in \mathbb C: - \pi/\theta < {\rm arg} \, z < \pi/\theta \}$
define
\begin{equation}\label{JT1}
\tilde{g}(z) = \int_0^L \log (z^\theta - y) \rho_{(1)}(y) \, dy, \qquad z \in H_\theta \backslash [0,\infty)
\end{equation}
and
\begin{equation}\label{JT2}
\tilde{G}(z) = {d \over d z} \tilde{g}(z).
\end{equation}

\begin{proposition} \cite[\S 4.5.1]{CR13}
Let $V(x) = x$ and set $L= (1 + \theta)^{1 + 1/\theta}$. Then the Green's function (\ref{JT2}), as it approaches the interval $[0,L]$
from the upper plane ($+$) or the lower plane ($-$) is given by
\begin{equation}
x \tilde{G}_\pm (x) = \theta ( I_\mp(x) + 1).
\end{equation}
\end{proposition}

Recalling (\ref{JT}) and (\ref{JI}) we thus have
\begin{equation}\label{6.6}
( x \tilde{G}_\pm(x)/\theta)^{1 + 1/\theta} {1 \over (x \tilde{G}_\pm(x)/\theta - 1)^{1/\theta} }=  x.
\end{equation}
Let us now replace $x$ by $x^{1/\theta}$ and write $x^{1/\theta} \tilde{G}_\pm( x^{1/\theta})/\theta = w$.
Minor manipulation of (\ref{6.6}) then gives
\begin{equation}
w^{\theta + 1} = x(w - 1),
\end{equation}
which we recognize as being identical to the equation (\ref{Gequation}) in the case $\theta + 1 = p$, $r=1$, uniquely determining the
Fuss-Catalan distribution. As a consequence, we have identified an equilibrium problem for this distribution.

\begin{corollary}
The energy functional
\begin{align}\label{AT}
E_\theta[\rho_{(1)}(y)] =&  \int_0^L
y^{1/\theta} \rho_{(1)}(y)  \, dy  \nonumber \\
&-{1 \over 2} \int_0^L dy \int_0^L dy' \,\rho_{(1)}(y) \rho_{(1)}(y') \log \Big (|y^{1/\theta}  - (y')^{1/\theta} |  | y - y'|  \Big ),
\end{align}
subject to the normalization $\int_0^L \rho_{(1)}(y) \, dy =1$  is minimized by the density function for the Raney distribution
$p=\theta + 1$, $r=1$, or equivalently Fuss-Catalan distribution with $s=\theta$.
\end{corollary}

 \begin{remark}
Changing variables $y= x^{\theta}$, the point process corresponding to (\ref{AT}) has its joint probability density function
proportional to
$$
\prod_{j=1}^N e^{- x_j} \prod_{1 \le j < k \le N} |x_j - x_k| |x_j^\theta - x_k^\theta|.
$$
This was introduced into random matrix theory by Muttalib \cite{Mu95} (for a realization as in terms of a Wishart matrix
formed out of triangular matrices, see the very recent work \cite{Ch14}), and the correlations at the hard edge subsequently
computed by Borodin \cite{Bo98}. Specifically, with Wright's generalized Bessel function defined as
$$
J_{a,b}(x) = \sum_{m=0}^\infty {(-x)^m \over m! \Gamma(a+bm)},
$$
it was shown in \cite{Bo98} that
$$
\rho_{(1)}(x) = \theta  \int_0^1 J_{1/\theta,1/\theta}(xt) J_{1,\theta}((xt)^\theta) \, dt.
$$
In keeping with Remark \ref{R2.6}, we expect a matching of the corresponding large $x$ form
with (\ref{2.16}) in the case $r=1$,
\begin{equation}
\rho_{(1)}(x) \, dx \Big |_{y= x^{\theta}} \mathop{\sim}\limits_{x \to \infty} {1 \over \pi}\sin  \Big (  {\pi \over \theta + 1}\Big ) y^{-1 + 1/(\theta + 1)} \, dy.
\end{equation}
\end{remark}

\medskip

We next turn our attention to the case $r=2$ of the Raney distributions. In the case $\theta = 2$, and using squared variables,
we know that the corresponding energy functional is given by (\ref{21.1}). In \cite[\S 4.5.2]{CR13} it is argued (without giving a
rigorous proof) that for the energy functional (\ref{21.3}) with potential $V(x) = (x - c)^2$, there is a value of $c$ for which the lower support
of the equilibrium
density goes from being positive to zero (or equivalently soft to hard edge), and for this value of $c$ the density is proportional to
$x^{(\theta - 1)/(\theta + 1)}$. Changing variables $y = x^{\theta}$ as done in arriving at (\ref{AT}), the density becomes proportional to
$y^{-(\theta - 1)/(\theta + 1)}$ which is in agreement with (\ref{2.16}) in the case $p=\theta + 1$, $r=2$.

\begin{remark}
Another way for the exponent $(\theta - 1)/(\theta + 1)$ to change sign is to replace $\theta$ by $1/\theta$.
With $0 < \theta < 1$, the sign changed exponent for the singular behaviour of a density specifed by an
energy functional can be found in the works of
Zinn-Justin \cite{ZJ00} and Kostov \cite{Ko00},  relating to the matrix model formalism of the six-vertex model on a random lattice.
With $\omega = e^{-i \pi \theta/2}$, $0 < \theta < 1$, the energy functional is
$$
E_\theta[\rho_{(1)}(y)] =  \int_0^L
(y - c)^2 \rho_{(1)}(y)  \, dy  
-{1 \over 2} \int_0^L dy \int_0^L dy' \,\rho_{(1)}(y) \rho_{(1)}(y') \log \Big | {y - y' \over \omega y + \omega^{-1} y'} \Big |,
$$
where $c$ is to be tuned so that the lower boundary of support goes from being positive to zero. In \cite{ZJ00}, the $n$-th moment of
$\rho_{(1)}(y)$, after suitable scaling of $y$, is computed to be equal to 
$$
2 {\Gamma (1 + (1 + \theta)n/2)  \Gamma (1 + (1 - \theta)n/2)  \over \Gamma(3 + n)}
$$
(cf.~(\ref{raneynumber})).
\end{remark}

\begin{remark}
This (heuristic) understanding of the Raney distribution for $r=2$ suggests that a realization of the general $r \in \mathbb Z^+$ case,
$r < p$, results from the equilibrium problem (\ref{21.3}) with $V(x)$ given by a certain ``tuned" degree $r$ polynomial in the
soft-to-hard transition (for ``tuned" polynomial potentials in the case of one-cut matrix models generalizing the Gaussian
ensembles --- so called critical unitary random matrix ensembles --- see \cite[Eq.~(1.8)]{CO11}). One would expect that it is
also necessary to change variables $y = x^{\theta}$, and identify $p=\theta + 1$ as required for $r=1$ and $r=2$.
\end{remark}

\begin{remark}
The equilibrium problem for the Raney distribution with parameters $p=3/2$, $r=1/2$ is known from \cite{sz}. Thus one knows that
the corresponding density $\rho_{(1)}(x)$ minimizes the energy functional
  \begin{align}\label{21.2a}
{E}_{3/2,1/2}[\rho_{(1)}(x)] = &
\int_0^L x \rho_{(1)}(x) \, dx   \nonumber \\
& -
{1 \over 2}
\int_0^L dx \int_0^L dx' \,\rho_{(1)}(x) \rho_{(1)}(x') \log \Big ({|x - x'| \over |x + x'|} |x^\theta - (x')^\theta|  \Big ),
\end{align}
with $\theta = 1$ and subject to the normalization $\int_0^L \rho_{(1)}(x) \, dx = 1$. An obvious question is to ask if
this same energy functional characterizes the Raney distribution with parameters $p=\theta/2 + 1$, $r=1/2$, and
if changing the linear potential $x$ in the first term to a specially tuned degree $k$ polynomial allows for $r$ to
be varied from $1/2$. The case $\theta = 2$ of (\ref{21.2a}) is the well known energy functional for the Marchenko-Pastur
law (\ref{MP1}), i.e.~the case $\theta = 1$ of (\ref{AT}). On the other hand the density for the Raney distribution with parameters
$p=2$, $r=1/2$ is given by (\ref{W212}). This agrees with (\ref{MP1}) upon the change of variables $y = x^{\theta}$, so as
found for  (\ref{21.3}), this change of variables will also be required as a final step.
 
\end{remark}

%%%%%%%%%%%Two new sections

 \section{A further class of polynomial equations}
 In the case $r=1$, $p > 1$ and integer,  the   general algebraic equation (\ref{Gequation}) specializes to the polynomial
 equation $w^p - z w + z = 0$. Here we will show that a random matrix structure recently considered in \cite{Fo14},
 involving the product of standard complex Gaussian matrices and their inverses, leads to a natural generalization of this
 equation for the resolvent.

For  $s, q\in \mathbb{N}_{0}=\{0,1,2,\ldots\}$, let $X_{1},\ldots, X_{s}, \widetilde{X}_{1},\ldots, \widetilde{X}_{q}$  be
independent standard complex Gaussian matrices.
Moreover,  let $X_{j}$ and $ \widetilde{X}_{k}$  be of dimension $N_{j}\times N_{j-1}$ ($j=1,2,\ldots, s$)
and $\widetilde{N}_{k}\times \widetilde{N}_{k-1}$ ($k=1,2,\ldots, q$) respectively.
Require that $N_0=\min\{N_0, \ldots, N_s\}$,  $\widetilde{N}_0=\min\{\widetilde{N}_0, \ldots, \widetilde{N}_q\}$, and furthermore that $N_0=\widetilde{N}_0=N$. With these specifications, introduce the products
$$Y_{s}=X_{s}X_{s-1} \cdots  X_{1}, \qquad \widetilde{Y}_{q}=\widetilde{X}_{q} \widetilde{X}_{q-1}\cdots \widetilde{X}_{1},$$
and use these to define product
Wishart-type  matrices involving inverse Gaussians according to
\be \label{inverseproduct}A_{s,q}=   \big(\widetilde{Y}_{q}^{*}\widetilde{Y}_{q}\big)^{-1/2} \big(Y_{s}^{*}Y_{s}\big) \big(\widetilde{Y}_{q}^{*}\widetilde{Y}_{q}\big)^{-1/2}.\ee

Another viewpoint  is to replace the rectangular matrix $\widetilde{X}_{k}$ by an $N\times N$ complex Gaussian matrix $\widetilde{X}_{k}$ with distribution proportional to $$\big(\det \widetilde{X}_{k}^{*}\widetilde{X}_{k}\big) ^{\mu_k} e^{-\textrm{tr}\widetilde{X}_{k}^{*}\widetilde{X}_{k}}, $$
where  $\mu_k= \widetilde{N}_{k}-N$.
Then $\tilde{Y}_q$ is a product of square matrices, and as in  \cite{Fo14} we can modify the definition (\ref{inverseproduct}) to a more familiar
Wishart form $((Y_s \tilde{Y}_q)^{-1})^* (Y_s \tilde{Y}_q^{-1})$.

For $q>0$, the averaged characteristic polynomial of $A_{s,q}$ is not well defined, since for example the averaged determinant diverges.
However, as noted in \cite{Fo14}, the eigenvalues of $A_{s,q}$ are the same as the eigenvalues in the generalized eigenvalue problem
$ Y_s^* {Y}_s \vec{v} =  \lambda \widetilde{Y}_{q}^{*}\widetilde{Y}_{q} \vec{v} $. The characteristic polynomial for the latter is
\be \label{Paverage}P^{(s,q)}_{N}(z):=\big\langle\det \big(z\widetilde{Y}_{q}^{*}\widetilde{Y}_{q}-Y_{s}^{*}Y_{s}\big)\big\rangle.\ee
%which equals $Q^{(s,q)}_{N}(z)$ if $\mu_{k}\rightarrow \mu_{k}-1$.
In \cite[Prop.~2]{Fo14} this has been evaluated, telling us that
\be P^{(s,q)}_{N}(z)=(-1)^{N}\prod_{l=1}^{s}(\nu_{l}+1)_{N}\, _{q+1}F_{s}\Big({-\widetilde{N}_{0},-\widetilde{N}_{1},\ldots,-\widetilde{N}_{q} \atop \nu_{1}+1, \ldots, \nu_{s}+1};(-1)^{q}z\Big),\ee
where $\nu_{l}=N_{l}-N$ ($l=1,2,\ldots,s)$.
Introducing the rescaled polynomial $$f(z)=P^{(s,q)}_{N}(\frac{N_{1}\cdots N_{s}}{ \widetilde{N}_{1} \cdots \widetilde{N}_{q}} z),$$
and recalling (\ref{DpFq})
we see  that $f$ satisfies the differential equation
\be \label{scaledpolynomial}\frac{(-1)^{q}}{\widetilde{N}_{1} \cdots \widetilde{N}_{q}}\prod_{j=0}^{q}\big(z\frac{d}{dz}-\widetilde{N}_{j}\big)f=\frac{1}{N_{1}\cdots N_{s}}\frac{d}{d z}\prod_{l=1}^{s}\big(z\frac{d}{d z}+\nu_{l}\big)f.\ee

As was the theme in Section \ref{S3}, we want to introduce the logarithmic derivative $g(z)= f'(z)/f(z)$, which we know from (\ref{3.4}) is
proportional to $N$ for $N$ large, and to take the $N \to \infty$ limit.
For this purpose, let us expand upon the reasoning which leads to (\ref{HN}).
With the aim being   to   express
$(z\frac{d}{dz})^{k}f$ in terms of $g, g', \ldots, g^{_{(k-1)}}$ and $f$, we will give a recursive definition of a polynomial $Q_{k-1}(g)$ in $g, g', \ldots, g^{_{(k-1)}}$
 with coefficients in $\mathbb{C}_{k-1}[z]=\langle 1, z, \ldots,z^{k-1}\rangle_{\mathbb{C}}$.

 First, let $Q_{0}=0$, then we have   \be (z\frac{d}{dz}) f=(zg+zQ_{0})f.\ee
 Secondly, suppose that $Q_{k-1}$ is well defined such that
 \be \label{kderivative} (z\frac{d}{dz})^{k} f=((zg)^{k}+zQ_{k-1})f,\ee
 and set \be \label{recursiverelation}Q_{k}=Q_{k-1}+z(g+\frac{d}{dz})Q_{k-1}+k(zg)^{k-1}(g+zg').\ee
 It is easy to verify that
\be \label{k+1derivative}(z\frac{d}{dz})^{k+1} f=((zg)^{k+1}+zQ_{k})f.\ee
 Moreover,   by induction we know from \eqref{recursiverelation}  that   the term  of highest degree of $Q_{k}$ in $g, g', \ldots, g^{_{(k-1)}}$ is \be \frac{k(k+1)}{2}(zg)^{k-1}(g+zg'),\ee
 which implies the degree of $Q_{k}$ is $k$.

  Let  $e_{k}(\nu_1,\ldots,\nu_s)$ be the $k$-th elementary symmetric polynomial in $\nu_1,\ldots,\nu_s$, as is consistent with (\ref{EY}).
  With \eqref{kderivative} and \eqref{k+1derivative} in mind, expanding both sides of (\ref{scaledpolynomial}) we have
  \begin{multline} \frac{(-1)^{q}}{\widetilde{N}_{1} \cdots \widetilde{N}_{q}}\sum_{j=0}^{q+1}\big((zg)^{j}+zQ_{j-1}(g)\big)e_{q+1-j}(-\widetilde{N}_{0},\ldots,-\widetilde{N}_{q})f\\
=\frac{1}{N_{1}\cdots N_{s}}\sum_{k=0}^{s}\big((zg)^{k}g+Q_{k}(g)\big)e_{s-k}(\nu_{1},\ldots,\nu_{s}) f. \end{multline} Here $Q_{-1}=0$ by convention.
Let $G_{N}(z)=\frac{1}{N}\frac{d}{d z}\log f(z),$  so that $g=NG_{N}$. Substituting it in the above and dividing both sides by $N$ show
\begin{multline}  -\prod_{j=0}^{q}\big(1-\frac{N}{\widetilde{N}_{j}} zG_{N}\big)+ \frac{(-1)^{q}}{\widetilde{N}_{0} \cdots \widetilde{N}_{q}}\sum_{j=0}^{q+1} zQ_{j-1}(N G_{N}) \,e_{q+1-j}(-\widetilde{N}_{0},\ldots,-\widetilde{N}_{q})\\
= G_{N}\prod_{k=1}^{s}\big(\frac{N}{N_{k}} zG_{N}+\frac{\nu_k}{N_k}\big)+\frac{1}{N N_{1}\cdots N_{s}}\sum_{k=0}^{s} Q_{k}(N G_{N})\,e_{s-k}(\nu_{1},\ldots,\nu_{s}). \end{multline}
Since the   degree of $Q_{k}$ is $k$, we can immediately deduce the following proposition.

\begin{proposition} Let $P^{(s,q)}_{N}(z)$ be given in \eqref{Paverage}.  Write  $$G_{N}(z)=\frac{1}{N}\frac{d}{d z}\log P^{(s,q)}_{N}\Big (\frac{N_{1}\cdots N_{s}}{ \widetilde{N}_{1} \cdots \widetilde{N}_{q}} z \Big ),$$
and $$G(z)=\lim_{N\rightarrow \infty}G_{N}(z).$$
Assume that as $N\rightarrow \infty$
\be \frac{N}{\widetilde{N}_{j}} \rightarrow u_{j}\in (0,1] \ \mbox{and} \ \frac{N}{N_{k}} \rightarrow v_{k}\in (0,1] \ \mbox{for}\  j=0, \ldots, q, \: k=1,\ldots,s, \ee
 where $u_0=1$. Then $G(z)$ satisfies  the polynomial equation
\be  \label{algebraicequation} -\prod_{j=0}^{q}\big(1-u_j zG(z)\big)= G(z) \prod_{k=1}^{s}\big(v_k zG(z)+1-v_{k}\big).
\ee \end{proposition}

\begin{remark}
The equation \eqref{algebraicequation} has been obtained for products of random matrices with general independent entries when $q=0$ \cite{agt2,agt3}, and also for products of Gaussian random matrices involving  inverses when all the $u_j$'s and $v_k$'s are equal  to 1  \cite{Fo14,HM12},  where the tools from free probability theory were used. Actually, the case that all the $u_j$'s are equal to 1 is special. Thus if
 $u_j<1$ for all $j=1, \ldots,q$, then the equation \eqref{algebraicequation}  has  one analytic solution  at infinity such that $zG(z)\rightarrow 1$ as $z\rightarrow \infty$  and thus the corresponding density has compact support. Specially, for $q=s=1$, i.e., the $F$ matrix case in statistics \cite{bs,Fo14}, an explicit density  has been given in Theorem 4.10, \cite{bs} and its limit of $u_1$ and $v_1$ approaching 1 is exactly the case of $s=q=1$ given in \cite{Fo14}. However, if $q>0$ and at least  one of $u_1, \ldots, u_q$ equals 1,  \eqref{algebraicequation}  has no analytic solution at infinity. Actually, we know from free probability theory that in this case the support of the density  is unbounded. To be precise, let $f_{s}(x)$ be the Fuss-Catalan distribution of degree s with support on $(0,K_{s})$,  and let $\tilde{f}_{q}(x)$ be the distribution of  $q$ inverse product of random matrices, then $\tilde{f}_{q}(x)=(1/x^{2}) f_{q}(1/x)$ with support on $(1/K_{q}, \infty)$.  Use of  tools from free probability shows that $f_{s,q}(x)$ is the multiplicative free convolution of $f_{s}(x)$ and $\tilde{f}_{q}(x)$, hence for $q>0$  the support of $f_{s,q}(x)$  is not compact. Furthermore, it is known from \cite[Prop.~6]{HM12} that the density is continuous.
\end{remark}

\begin{remark}
In the case $q=0$, with each $N_j - N = \nu_j$ fixed in the limit $N \to \infty$, differential equations have been shown to also characterize different
observable quantities, namely the gap probabilities at the hard edge. Moreover, these differential equations are nonlinear, and in fact related
to isomonodromy preserving deformations of linear systems \cite{St14}
\end{remark}

 \section{Further development  of the parameterization method}
We start from the polynomial equation \eqref{algebraicequation} satisfied by the Stieltjes transform of  the inverse product \eqref{inverseproduct} and
use   the  parameterization method of the spectral variable, independently  due to Biane and Neuschel, to
give  an explicit form of the limiting density in a special case  when all the $u_j$'s $(j=1,\dots,q)$ and $v_k$'s $(k=1,\dots, s)$ are equal  to 1.
Recently, this same task has been undertaken by Haagerup and M\"oller \cite{HM12} using the strategy of Biane \cite{Ba99}.
We remark too that
 in the case of  $q=s$ the exact form of the density has been computed in
 \cite{Ba98} and independently in \cite{Fo14} without the use of a parametrized spectral variable,
 giving
\be \label{ssdensity} f_{s,s}(x)=\frac{1}{\pi}\frac{x^{-s/(s+1)}\sin  \pi/(s+1)}{ 1+2x^{1/(s+1)}\cos \pi/(s+1)+x^{ 2/(s+1)}}, \qquad 0<x<\infty.\ee
As an application, the  $x \rightarrow 0^+$ leading asymptotics of  the density can be read off to be equal to $\frac{1}{\pi}x^{-\frac{s}{1+s}}\sin\frac{1}{1+s}\pi$.
It  was noticed in  \cite{Fo14} that this is the same form as that for $q=0$, as deduced from the result
(\ref{1.7}) (take $r=1$, $p=s+1$ in the first case of (\ref{2.16})),
 from which it was conjectured that this should be a universal feature valid for general $s, q$ but independent of $q$.    As an application of our
 explicit determination of the density for general $s \ne q$, we are able to obtain the corresponding $x \rightarrow 0^+$ leading asymptotic form,
 and so give an affirmative answer to this conjecture.
 \subsection{Explicit densities}\label{subsectsq}
\subsubsection{General procedure}
Assume that the Stieltjes transform $$G(z)=\int_0^\infty \frac{1}{z-x}f_{s,q}(x) dx$$ of the density $f_{s,q}(x)$ satisfies the equation
\be \label{sqtransform}(zG)^{1+s} +z(1-zG)^{1+q}=0.\ee
 This equation is a limit case of \eqref{algebraicequation} when all $u_i$ and $v_j$ approach  1 from the below. Then the support changes from the compact case to the noncompact one (more
 precisely, as remarked above this happens as soon was one of the $u_j$ becomes equal to 1).
Let $w(z)=zG(z)$. We will try to find two special solutions of  the algebraic equation
  \be \label{sqwequation}w^{1+s} +z(1-w)^{1+q}=0,\ee
one as the cut $0\le z < \infty$ is approached from Im$(z)>0$, the other as it is approached from Im$(z) < 0$,
from which the density immediately follows.

Note that \eqref{sqwequation} can be rewritten as
 \be w^{\frac{1+s}{1+q}} +(-(-z)^{\frac{1}{1+q}})(1-w)=0. \ee If we treat $w$ as a function of the new variable $\hat{z}=-(-z)^{1/(1+q)}$, then \eqref{sqwequation} is just the equation \eqref{Gequation} with $p=(1+s)/(1+q)$, $r=1$ and $z$ is substituted by $\hat{z}$.
 We know that in the variable $\hat{z}$ this has the solution about infinity given by the series in (\ref{1.2a}). But this is not analytic about infinity in the variable $z$ when $q>0$,
 and so the support must be unbounded.

Following Neuschel's strategy, and augmenting this by explicit knowledge of the structure of the parameterization (\ref{1.5}),  we
begin  by seeking a  complex conjugate
pair of solutions of (\ref{sqwequation}) in the  parameterized
polar coordinates  form
\be \label{w=}w_{\pm}=\frac{\sin(a\varphi +\varphi+b)}{\sin( a\varphi)}\, e^{\pm i(\varphi+b)}, \ee
where $a$ and $b$ are to be determined.  Note that this has the property that $w_\pm = 1$ for $\varphi = - b$.
Simple manipulation then gives
\be \label{1-w=}1-w_{\pm}=-\frac{\sin(\varphi+b)}{\sin( a\varphi)}\, e^{\pm i(a\varphi +\varphi+b)}.\ee
Substituting \eqref{w=} and \eqref{1-w=} in \eqref{sqwequation} one establishes the corresponding parameterization of $z$,
\be z=\frac{ (\sin(a\varphi +\varphi+b) )^{1+s}}{ (\sin(\varphi+b) )^{1+q} (\sin( a\varphi) )^{s-q}}\, e^{\mp i (a(1+q)\varphi+q\pi -(s-q)(\varphi+b) )}.\ee

To ensure that $z$ lies in one cut of the real axis we suppose the phase satisfies
$$a(1+q)\varphi+q\pi -(s-q)(\varphi+b)\equiv 2k\pi$$ for some  suitably chosen  $ k\in \mathbb{Z}$. Therefore we get
$$a(1+q)   -(s-q)=0, \qquad q\pi- (s-q)b= 2k\pi,$$
and thus
\be a=\frac{s-q}{1+q}, \qquad b=\frac{2k-q}{q-s}\pi\ (s\neq q).\ee
Supposing that $s\neq q$, it follows from this working that  if we use the parameterization
\be\label{sqpar} x=\rho(\varphi)=\frac{ (\sin(\frac{1+s}{1+q}\varphi +\frac{2k-q}{q-s}\pi) )^{1+s}}{ (\sin(\varphi+\frac{2k-q}{q-s}\pi) )^{1+q}\, (\sin( \frac{s-q}{1+q}\varphi) )^{s-q}},\ee then the complex conjugate pair  of solutions of \eqref{sqwequation} is given by
\be \label{sqsolution} \frac{ \sin(\frac{1+s}{1+q}\varphi +\frac{2k-q}{q-s}\pi) }{  \sin(\frac{s-q}{1+q}\varphi)  }   \, e^{\pm i(\varphi+\frac{2k-q}{q-s}\pi)}.\ee

According to  the inverse formula of the Stieltjes transform  the density function,  $f_{s,q}(x)$ say, is given by
 $$ f_{s,q}(x)=\lim_{\epsilon\rightarrow 0+}\frac{1}{2i\pi}\Big(\frac{w(x-i\epsilon)}{x-i\epsilon}-\frac{w(x+i\epsilon)}{x+i\epsilon}\Big),$$
 where $w$ is one of the above two solutions. Note that  the imaginary parts of $w(z)$ and $z$ have opposite sign;  we choose $w=w_-$ for $s>q$  while $w=w_+$ for $s<q$. We then have  for  $s\neq q$
 \begin{align} \label{sqdensity}f_{s,q}(\rho(\varphi))&= \frac{1}{\pi \rho(\varphi)}\frac{ \sin(\frac{1+s}{1+q}\varphi +\frac{2k-q}{q-s}\pi) }{\mbox{sign}(s-q) \sin(\frac{s-q}{1+q}\varphi)  }   \,  \sin(\varphi+\frac{2k-q}{q-s}\pi) \nonumber\\
&=\frac{1}{\pi}\frac{ (\sin(\varphi+\frac{2k-q}{q-s}\pi))^{2+q} }{\mbox{sign}(s-q) (\sin(\frac{1+s}{1+q}\varphi +\frac{2k-q}{q-s}\pi))^{s}}   \, (\sin(\frac{s-q}{1+q}\varphi))^{s-q-1}.\end{align}
The remaining tasks are the determination of
$k$ as well as  the range of $\varphi$.

\subsubsection{Case  $s>q$} To ensure that the right-hand sides of   \eqref{sqpar} and \eqref{sqdensity} are nonnegative  we must choose an appropriate $k$ and  restrict  the range of $\varphi$.  First, a restriction following from the periodicity of the sine functions is $0<\varphi+\frac{2k-q}{q-s}\pi<\pi$. Second, the nonnegativity of the density suggests that both $ \frac{1+s}{1+q}\varphi+\frac{2k-q}{q-s}\pi$ and $\frac{s-q}{1+q}\varphi$ should belong to $(2l\pi, (2l+1)\pi)$ for some $l\in \mathbb{Z}$. So we can choose $k=q$ for convenience and thus get the range
\be \frac{q}{s-q}\pi<\varphi<\frac{s}{1+s}\frac{1+q}{s-q}\pi.\ee

The final form now follows. This is stated in Proposition \ref{sqdensity} below, where for convenience  $\varphi$ has been replaced by $\varphi+\frac{q}{s-q}\pi$.

\subsubsection{Case  $s<q$}
In this case we rewrite \eqref{sqpar}  as
\be  x=\rho(\varphi)=\frac{ (\sin(\frac{q-2k}{q-s}\pi-\frac{1+s}{1+q}\varphi) )^{1+s}}{ (\sin(\frac{q-2k}{q-s}\pi-\varphi) )^{1+q}\, (\sin( \frac{q-s}{1+q}\varphi) )^{s-q}},\ee
and choose $k=0$.   To ensure that the angles of the sine functions above  lie  in the interval $(0, \pi)$ we must restrict $\varphi$ to the range
\be \frac{s}{1+s}\frac{1+q}{q-s}\pi <\varphi<\frac{q}{q-s}\pi.\ee

 The sought form of the parameterization is specified in Proposition \ref{sqdensity}, where for convenience
$\varphi$ has been replaced by $\frac{q\pi}{q-s}-\varphi$, along with the case $s>q$ of the previous subsection, and the case $s=q$ which can be
checked separately (cf.\eqref{ssdensity}).

\begin{proposition}\label{sqdensity}
Assume   $s,q\geq 0$ and $(s,q)\neq (0,0)$.
If we   use the parameterization (a strictly  decreasing function)
  \be  x=\rho(\varphi)=\frac{ (\sin(\frac{1+s}{1+q}\varphi+\frac{q\pi}{1+q} ) )^{1+s}}{ (\sin  \varphi )^{1+q}\, (\sin( \frac{s-q}{1+q}\varphi+\frac{q\pi}{1+q}) )^{s-q}}, \qquad 0<\varphi<\frac{\pi}{1+s}, \ee
   then
 \begin{align}  f_{s,q}(\rho(\varphi))
&=\frac{1}{\pi}\frac{ (\sin(\frac{s-q}{1+q}\varphi+\frac{q\pi}{1+q}))^{s-q-1} }{ (\sin(\frac{1+s}{1+q}\varphi +\frac{q\pi}{1+q} ))^{s}}   \,  (\sin \varphi)^{2+q}.\end{align}
\end{proposition}

A direct application  of
Proposition \ref{sqdensity}  gives the explicit leading
  asymptotic form of the density upon the approach of either boundary of its support.

\begin{corollary} \label{s>qcoro} Let  $s, q>0$. We have, for $x \rightarrow 0^+$
\be
f_{s,q}(x) \sim   \frac{1}{\pi}\sin\frac{\pi}{1+s} \, x^{-\frac{s}{1+s}} ,
\ee
while for $x\rightarrow \infty$
 \be
f_{s,q}(x) \sim   \frac{1}{\pi}\sin\frac{\pi}{1+q} \, x^{-\frac{2+q}{1+q}}.
\ee
\end{corollary}

\begin{proof} Noting that   $x \to 0 \   \mbox{as} \  \varphi \to  \frac{\pi}{1+s},$
and  $x \to \infty \   \mbox{as} \  \varphi \to   0$, % for $sq\neq 0$,
%(similarly for the case $sq=0$)
 a
simple computation  completes the proof.
\end{proof}

Actually, for general $s, q\geq 0$ we have all leading asymptotics as follows:
(i) for $s>0$ and $q\geq 0$,   the leading form is $\frac{1}{\pi}\sin(\pi/(1+s)) x^{- s/(1+s)}$ as $x \rightarrow 0^+$;
(ii) for $q>0$ and $s\geq 0$,   the leading form is $\frac{1}{\pi}\sin(\pi/(1+q)) x^{- (2+q/)(1+q)}$ as $x \rightarrow \infty$. This latter behavior is consistent with the fact that all moments diverge.  We remark too, as proved respectively from the Stieltjes transform and the $S$- transform in \cite{Fo14} and \cite{HM12}, that there is a duality relation between the densities, being unchanged by the mappings
\be \label{dualitylaws}s\longleftrightarrow q, \qquad x f_{s,q}(x)\longrightarrow x f_{q,s}(x), \qquad x\longrightarrow \frac{1}{x}.\ee

\begin{remark} Recently, Haagerup and M\"{o}ller  have  proved the same results as in Proposition \ref{sqdensity}, see \cite[Theorem 6]{HM12}. They obtained the parametrization representation by studying   the free multiplicative
convolution and the $S$-transform, while our starting point is the Stieltjes transform and the related equation \eqref{sqtransform}. We will also give  direct  expression of densities in terms of spectral variables in two special cases in the subsequent subsection.
\end{remark}

\subsection{Two special cases}

 In this subsection we  discuss  the special case  $1+s=2(1+q)$  or $1+q=2(1+s)$, and give an  explicit form of the density in the original spectral variable,
 analogous to the expression (\ref{ssdensity}) in the case $1+s = 1+q$.
Inspection of  \eqref{sqwequation} shows that these two special cases give a quadratic equation in $w$, which permits further analysis.

 Consider first the case $1+s=2(1+q)$ or equivalently $s=1+2q$. The quadratic equation then reads
   \be \label{qd}
    w^2+(-z)^{1/(1+q)}w-(-z)^{1/(1+q)}=0,\ee
and we read off for the roots
\be \label{sq}
w_{\pm}=\frac{1}{2}\big(-(-z)^{1/(1+q)}\pm \sqrt{(-z)^{2/(1+q)}+4(-z)^{1/(1+q)}}\,\big).
\ee
Here the square root is specified as the one with the positive imaginary part. Note that we require $w(z)\rightarrow 1$ as $z\rightarrow -\infty$, so we choose $w_+$ as
the solution corresponding to the Stieltjes transform. From this we compute the density
\begin{align*} x f_{s,q}(x)&=\frac{1}{\pi}\lim_{\epsilon\rightarrow 0^+}\mbox{Im}\, w_{+}(x-i\epsilon)\\
&=\frac{1}{2\pi} \mbox{Im}\, \big\{-(x e^{i\pi})^{1/(1+q)}+\sqrt{(x e^{i\pi})^{2/(1+q)}+4(x e^{i\pi})^{1/(1+q)}}\,\big\}.\end{align*}
 To take the imaginary part, we notice that if we set for $q>0$ (the case $q=0$ is just the Marchenko-Pastur law)
 $$1+4 x^{-1/(1+q)}  e^{-i\pi/(1+q)}=Re^{-i\theta}, \qquad 0<\theta< \pi/(1+q),$$
 where the positive number $R$ satisfies
 \be
 R^2=1+16x^{-2/(1+q)}+8x^{-1/(1+q)}\cos\tfrac{\pi}{1+q},
\ee
then
\begin{align}  \label{fd1}
x f_{s,q}(x)&=\frac{1}{2\pi} x^{1/(1+q)} \big(-\sin\tfrac{\pi}{1+q} +\sqrt{R}\sin(\tfrac{\pi}{1+q}-\tfrac{\theta}{2})\big)\nonumber\\
&=\frac{1}{2\pi} x^{1/(1+q)} \Big(-\sin\tfrac{\pi}{1+q} +\sqrt{\tfrac{R+1}{2}+2x^{-1/(1+q)}\cos\tfrac{\pi}{1+q}}\, \sin\tfrac{\pi}{1+q}\nonumber \\
&- \sqrt{\tfrac{R-1}{2}-2x^{-1/(1+q)}\cos\tfrac{\pi}{1+q}}\, \cos\tfrac{\pi}{1+q} \,\Big).\end{align}

 The case $1+q=2(1+s)$ or equivalently $q=1+2s$ is similar. We obtain for $s>0$
 \begin{align}   \label{fd2}
 x f_{s,q}(x)%&=\frac{1}{2\pi} x^{-1/(1+s)} \big(-\sin\tfrac{\pi}{1+s} +\sqrt{\tilde{r}}\sin(\tfrac{\pi}{1+s}-\tfrac{\theta}{2})\big)\\
&=\frac{1}{2\pi} x^{-1/(1+s)} \Big(-\sin\tfrac{\pi}{1+s} +\sqrt{\tfrac{\tilde{R}+1}{2}+2x^{1/(1+s)}\cos\tfrac{\pi}{1+s}}\, \sin\tfrac{\pi}{1+s}\nonumber \\
&- \sqrt{\tfrac{\tilde{R}-1}{2}-2x^{1/(1+s)}\cos\tfrac{\pi}{1+s}}\, \cos\tfrac{\pi}{1+s} \,\Big),\end{align}
 where   the positive number $\tilde{R}$ satisfies\be
 \tilde{R}^2=1+16x^{2/(1+s)}+8x^{1/(1+s)}\cos\tfrac{\pi}{1+s}.
\ee

\begin{remark}
Careful computations using (\ref{fd1}) and  (\ref{fd2})
give the same asymptotic behaviours of densities for $x\to 0^+$ and $x \to \infty$ obtained in the previous subsections.
\end{remark}

\begin{remark}
It is of interest to note how the above working relates to the parameterization approach. In the latter, with $w$ parameterized according to
(\ref{w=}), the variable $(-z)^{1/(1+q)} $ is written
$(-z)^{1/(1+q)} = w^2/(1 - w)$ as is consistent with (\ref{qd}). The terms inside the square root of (\ref{sq}) can then be
written as a perfect square, thus eliminating the square root and providing simplification.
\end{remark}
\subsection{Densities from  mixed equations}
 Our extension of the parameterization method  is also applicable to the more general equation with any  $s, q \geq 0$ \be \label{sqrwequation}w^{\frac{1+s}{r}} +z(1-w^{\frac{1}{r}})^{1+q}=0, \qquad 0<r\leq  1+s,\ee
which is a mixed case  of equations \eqref{Gequation} and \eqref{sqwequation}. Here   the Stieltjes transform $$G(z)=\int_0^\infty \frac{1}{z-x}f_{s,q,r}(x) dx$$ of the density $f_{s,q,r}(x)$ satisfies \eqref{sqrwequation} with $w=zG(z)$.

   \subsubsection{$s=q$}
    In this case we get from \eqref{sqrwequation} that
   \be w=\left(\frac{(-z)^{1/(1+s)}}{1+(-z)^{1/(1+s)}}\right)^{r}.\ee
   Thus, \begin{align}
x f_{s,s,r}(x)&=\frac{1}{2\pi i} \lim_{\epsilon\rightarrow 0^+}\big(\textrm{Im}\, w(x-i\epsilon )-\textrm{Im}\, w(x-i\epsilon )\big)\nonumber\\
&=\frac{x^{r/(1+s)}}{2\pi i} \frac{\big( x^{1/(1+s)}+e^{i\pi/(1+s)} \big)^{r}-\big( x^{1/(1+s)}+e^{-i\pi/(1+s)} \big)^{r}}{\big( 1+2x^{1/(1+s)}\cos\tfrac{\pi}{1+s}+x^{2/(1+s)} \big)^{r}}. \end{align}

Furthermore, for $x>0$, if we let $$x^{1/(1+s)}+e^{i\pi/(1+s)}=R e^{i\varphi},\qquad 0<\varphi<\tfrac{\pi}{1+s},$$
where
\be R=\sqrt{1+2x^{1/(1+s)}\cos\tfrac{\pi}{1+s}+x^{2/(1+s)}},\ee
then we have
\be \label{s=qrdensity}x f_{s,s,r}(x)=\frac{1}{\pi}\frac{x^{r/(1+s)}}{R^{r}}\sin(r\varphi).\ee

\subsubsection{$s\neq q$}
 As in the subsection \ref{subsectsq}, if we use the parameterization
\be  x=\rho(\varphi)=\frac{ (\sin(\frac{1+s}{1+q}\varphi +\frac{2k-q}{q-s}\pi) )^{1+s}}{ (\sin(\varphi+\frac{2k-q}{q-s}\pi) )^{1+q}\, (\sin( \frac{s-q}{1+q}\varphi) )^{s-q}},\ee then the complex conjugate pair  of solutions of \eqref{sqrwequation} is given by
\be \label{sqrsolution} w_{\pm}=\Big(\frac{ \sin(\frac{1+s}{1+q}\varphi +\frac{2k-q}{q-s}\pi) }{  \sin(\frac{s-q}{1+q}\varphi)  }   \, e^{\pm i(\varphi+\frac{2k-q}{q-s}\pi)}\Big)^{r}.\ee

Choose $k=q, w=w_-$ for $s>q$  while $k=0, w=w_+$ for $s<q$, after similar discussion in the subsection \ref{subsectsq} we    have  the following proposition including \eqref{s=qrdensity}.
\begin{proposition}\label{sqrdensity}
Assume   $s,q\geq 0$, $(s,q)\neq (0,0)$ and $0<r\leq 1+s$.
If we   use the parameterization
  \be  x=\rho(\varphi)=\frac{ (\sin(\frac{1+s}{1+q}\varphi+\frac{q\pi}{1+q} ) )^{1+s}}{ (\sin  \varphi )^{1+q}\, (\sin( \frac{s-q}{1+q}\varphi+\frac{q\pi}{1+q}) )^{s-q}}, \qquad 0<\varphi<\frac{\pi}{1+s}, \ee
   then
 \begin{align}  f_{s,q,r}(\rho(\varphi))
&=\frac{1}{\pi}\frac{(\sin(\frac{s-q}{1+q}\varphi+\frac{q\pi}{1+q}))^{s-q-r}  }{ (\sin(\frac{1+s}{1+q}\varphi +\frac{q\pi}{1+q} ))^{1+s-r}}   \,  (\sin \varphi)^{1+q} \sin(r\varphi).\end{align}
\end{proposition}

A corollary immediately follows from
Proposition \ref{sqrdensity}.

\begin{corollary} \label{s>qcoro} Let  $s, q>0$. We have, for $x \rightarrow 0^+$
\be
f_{s,q,r}(x) \sim  \begin{cases} \frac{1}{\pi}\sin\frac{\pi}{1+s} \, x^{ \tfrac{1}{1+s}}, \qquad &r=1+s;\\
\frac{1}{\pi} \sin\frac{r\pi}{1+s} \,  x^{-1+ \tfrac{r}{1+s}}\ , \qquad &r<1+s,
\end{cases}
\ee
while for $x\rightarrow \infty$
 \be
f_{s,q,r}(x) \sim   \frac{r}{\pi} \sin\tfrac{\pi}{1+q} \, x^{-\tfrac{2+q}{1+q}}.
\ee
\end{corollary}

\begin{remark} %The equation \eqref{sqrwequation} is unchanged by the mappings  of
%\be s\leftrightarrow q, \qquad zG(z)\mapsto  \big(1-((1/z)G(1/z))^{1/r}\big)^{r}, \qquad z \mapsto 1/z,\ee
%which implies the duality laws in \eqref{dualitylaws} when $r=1$. We also stress that
 The equation  \eqref{sqrwequation} with some special $r$ may occur in the products and inverses of random matrices, as in Remarks  \ref{remarkcritical} and \ref{remarkbures}.\end{remark}

\begin{acknow}
The work of P.J.~Forrester was supported by the Australian Research Council, for the project `Characteristic polynomials in random matrix
 theory'.  The work of D.-Z.~Liu  was  supported by the National Natural Science Foundation of China under grants  11301499 and  11171005, and by CUSF WK 0010000026. He  would also like to express his sincere thanks to the Department of Mathematics and Statistics, The University of Melbourne for its hospitality during his stay. We thank Arno Kuijlaars and Thorsten Neuschel for altering us to the work of Haagerup and M\"oller, and Lun Zhang for helpful comments on the first draft.

\end{acknow}

\bibliographystyle{amsplain}

\end{document}